\documentclass[letterpaper, 10pt, twocolumn]{article}
\usepackage{titlesec}
\titleformat{\section}{\normalfont\large\bfseries}{\thesection}{1em}{} % Section headings: slightly larger than body
\titleformat{\subsection}{\normalfont\normalsize\bfseries}{\thesubsection}{1em}{} % Subsection headings: body size, bold
\titleformat{\subsubsection}{\normalfont\small\bfseries}{\thesubsubsection}{1em}{} % Subsubsection headings: a bit smaller
\usepackage[letterpaper,top=0.75in,bottom=1in,left=0.625in,right=0.625in]{geometry}
\usepackage{cite}
\usepackage{amsmath,amssymb,amsfonts}
\usepackage{graphicx}
\usepackage{textcomp}

\usepackage{mathtools}
\usepackage{relsize}
\usepackage{tikz}
\usetikzlibrary{positioning, arrows.meta}

\usepackage{tabularx} % extra features for tabular environment
\usepackage{graphicx} % takes care of graphic including machinery
\usepackage{cite} % takes care of citations
\usepackage[final]{hyperref} % adds hyper links inside the generated pdf file
\usepackage{amsmath} % assumes amsmath package installed
\usepackage{amssymb}  % assumes amsmath package installed
\usepackage{mathtools, cuted,bm}
\usepackage{etoolbox}
\patchcmd{\proof}{\indent}{}{}{}
\usepackage{caption} % For better caption control
\usepackage{tabularray}

\usepackage{tabularx}
\usepackage{array}
\newcolumntype{Y}{>{\centering\arraybackslash}X}

\usepackage[justification=centering]{caption}

\newcommand{\twodots}{.\kern-0.1em.}
\newcommand{\myldots}{\kern-0.05em.\kern-0.01em.\kern-0.01em.\kern0.01em}

\usepackage{caption}
\usepackage{subcaption}
\usepackage{esvect}
\usepackage{lipsum, color}

\newcommand{\R}{\mathbb{R}}

\newcommand{\fe}{\mathsf{f}}
\newcommand{\ve}{\mathsf{v}}

\usepackage{relsize}
\usepackage{tikz}
\graphicspath{{./Figures/}} 

\usepackage[margin=0.8cm]{caption}
\usepackage{pgfplots}
\pgfplotsset{compat=newest}
\usepackage{tabularx}
\pgfplotsset{plot coordinates/math parser=false}

\usepackage{ulem}
\usepackage{cancel}

\usepackage{booktabs}

\makeatletter
\def\munderbar#1{\underline{\sbox\tw@{$#1$}\dp\tw@\z@\box\tw@}}
\makeatother

\definecolor{olivegreen}{RGB}{128, 128, 0}

\usepackage{algpseudocode}
\usepackage{algorithm}

\usepackage{caption}

\newcommand\scalemath[2]{\scalebox{#1}{\mbox{\ensuremath{\displaystyle #2}}}}

\hypersetup{
	colorlinks=true,       % false: boxed links; true: colored links
	linkcolor=black,        % color of internal links
	citecolor=black,        % color of links to bibliography
	filecolor=magenta,     % color of file links
	urlcolor=black         
}

\usepackage{amsthm}
\newtheorem{remark}{Remark}
\newtheorem{problem}{Problem}
\newtheorem{proposition}{Proposition}

\newtheorem{assumption}{Assumption}

\begin{document}
	%\title{LPV System Identification with Constrained Controller Synthesis Guarantees}
	\title{\textbf{Combined Learning of Linear Parameter-Varying Models and Robust Control Invariant Sets}}
	%\author{Sampath Kumar Mulagaleti and Alberto Bemporad, \IEEEmembership{Fellow, IEEE}
		%\thanks{ }\thanks{}}
	
	\author{Sampath Kumar Mulagaleti and Alberto Bemporad, \textit{Fellow, IEEE}
		\thanks{The authors are with the IMT School for Advanced Studies
			Lucca, Piazza San Francesco 19, 55100, Lucca, Italy (\url{s.mulagaleti,
				alberto.bemporad}@imtlucca.it). This work was supported by the European Research Council (ERC), Advanced Research Grant COMPACT (Grant Agreement No. 101141351).}
	}

	%\title{
		%LPV System identification with constrained controller synthesis guarantees
		%}
	%\author{Sampath Kumar Mulagaleti, Alberto Bemporad}

	\newcounter{tempEquationCounter}
	\newcounter{thisEquationNumber}
	\newenvironment{floatEq}
	{\setcounter{thisEquationNumber}{\value{equation}}\addtocounter{equation}{1}% record equation as happened and remember number
		\begin{figure*}[!t]% float following equation across columns
			\normalsize\setcounter{tempEquationCounter}{\value{equation}}% record current equation number in floated location
			\setcounter{equation}{\value{thisEquationNumber}}% use previous equation number
		}
		{\setcounter{equation}{\value{tempEquationCounter}}% set back to equation number in floated location
			\hrulefill\vspace*{4pt}% add a horizontal rule separator
		\end{figure*}% end float environment
		
	}
	\newenvironment{floatEq2}
	{\setcounter{thisEquationNumber}{\value{equation}}\addtocounter{equation}{1}% record equation as happened and remember number
		\begin{figure*}[!t]% float following equation across columns
			\normalsize\setcounter{tempEquationCounter}{\value{equation}}% record current equation number in floated location
			\setcounter{equation}{\value{thisEquationNumber}}% use previous equation number
		}
		{\setcounter{equation}{\value{tempEquationCounter}}% set back to equation number in floated location
			%\hrulefill\vspace*{4pt}% add a horizontal rule separator
		\end{figure*}% end float environment
	}

	\maketitle
	
	\begin{abstract}
		Dynamical models identified from data are frequently employed in control system design. However, decoupling system identification from controller synthesis can result in situations where no suitable controller exists after a model has been identified. In this work, we introduce a novel control-oriented regularization in the identification procedure to ensure the existence of a controller that can enforce constraints on system variables robustly. The combined identification algorithm includes: (i) the concurrent learning of an uncertain model and a nominal model using an observer; (ii) a regularization term on the model parameters defined as the size of the largest robust control invariant set for the uncertain model. To make the learning problem tractable, we consider nonlinear models in quasi Linear Parameter-Varying (qLPV) form, utilizing a novel scheduling function parameterization that facilitates the derivation of an associated uncertain linear model. The robust control invariant set is represented as a polytope, and we adopt novel results from polytope geometry to derive the regularization function as the value of an optimization problem. Additionally, we present new model-reduction approaches that exploit the chosen model structure. Numerical examples on classical identification benchmarks demonstrate the efficacy of our approach. A simple control scheme is also derived to provide an example of data-driven control of a constrained nonlinear system.
	\end{abstract}
	
	\thispagestyle{empty}
	\pagestyle{empty}
	
	\section{Introduction}
	As cyber-physical systems become increasingly prevalent, there is a growing need for algorithms that can safely and effectively control them. Model-based control design, which relies on dynamical models learned from data through system identification (SysID) techniques, has proven successful in addressing this need. These models predict future system behavior, and controllers synthesized based on them optimize the predictions to achieve desired control objectives when used in closed-loop with the system. A typical control objective is to perform reference tracking while satisfying input and output constraints.
	Typically, SysID and controller synthesis are performed sequentially: first, get a model from data, then design a controller for it. This approach, however, can result in controller parameterizations that are suboptimal or even infeasible for a given model parameterization~\cite{VanDenHof1995}. For instance, a nominal controller\,---\,designed without accounting for modeling errors\,---\,may lead to constraint violations when applied to the plant generating the data. Alternatively, a robust controller to explicitly consider model mismatch might be impossible to synthesize, i.e., there might not exist a controller of a given parameterization for the identified uncertain model. Consequently, there is a need for methodologies that integrate SysID and controller synthesis into a unified framework.
	
	Existing unifying approaches can broadly be divided into two categories: 
	
	$i$) Identification of a model and its associated uncertainty, followed by robust controller synthesis based on the identified uncertain model. Some related contributions are~\cite{kosut1992set,Milanese2004,sadraddini2018formal,Terzi2019,Lauricella2020}, etc., which focus on set-membership and/or robust identification;  
	
	$ii$) Concurrent uncertain system identification and robust controller synthesis, e.g.,~\cite{Shook1992,Chen2018,Mulagaleti2022}, etc., in which the model, uncertainty, and controller parameterizations are fixed, and parameters are optimized to minimize a control objective while enforcing model unfalsifiability. Recent reinforcement learning-based approaches, notably,~\cite{Zanon2021}, can also be interpreted in this framework, in which the parameters of a linear uncertain model and a robust Model Predictive Controller (MPC) are optimized. These approaches guarantee the existence of a controller (by directly identifying it) for the uncertain model, which under reasonable assumptions is a constraint-satisfying controller for the underlying plant. However, models identified using such methods might perform poorly when used for the synthesis of control schemes other than the one they are identified for. While incorporating prediction loss into the control objective can mitigate this issue, the methods are typically designed for input-state datasets rather than input-output datasets, requiring further modifications.
	
	Assuming there exists a method to characterize the uncertainty associated with a given model, the key step towards ensuring that one can design a model-based controller satisfying constraints involves guaranteeing the existence of a robust control invariant (RCI) set~\cite{Blanchini2008} for that model. RCI sets are regions of the state space in which the uncertain model can be regulated indefinitely within output constraints using feasible control inputs, and computational methods to construct them mostly rely on
	linear, possibly uncertain, models. Therefore, when dealing with nonlinear systems, it is reasonable to consider a linear uncertain model whose trajectories include those of the nonlinear model of the system; for this reason, Linear Parameter-Varying (LPV) model classes have been considered, and in particular quasi LPV (qLPV) models, in which the scheduling parameter is a nonlinear function
	of the state and input vectors. Data-driven methods for LPV system identification in the prediction error minimization paradigm have been developed under assumptions of availability of the scheduling function~\cite{Bamieh1999}. Avoiding this assumption, a large number of works have been dedicated to (implicitly) estimating the state sequence from input-output data. In~\cite{Masti2021}, an autoencoder was used for defining the hidden states of a nonlinear model, such as a qLPV model (see \cite[Sect.~4.2]{Masti2021}). In~\cite{Verhoek2022}, a joint identification procedure of the scheduling function and system matrices was presented, in which intermittent states were estimated by an encoder network~\cite{Beintema2021}. 
	
	\subsection{Contributions}
	We present a SysID framework to identify an uncertain model of the plant, in which we introduce a \textit{control-oriented} regularization that guarantees the existence of a constraint-satisfying RCI set for the model being identified. We parameterize a nonlinear model in qLPV form and the associated RCI set as a polytope.  We then formulate the control-oriented regularization as the value function of an {\color{black} optimization problem}, the goal of which is to {\color{black} compute an RCI set with minimal conservativeness.} Using every model such that the objective of the regularized SysID problem is finite, we can design a constraint-satisfying controller for the underlying plant.
	Our interpretation of control-oriented regularization is fundamentally different from that in~\cite{formentin2021control}, in which it refers to a function that minimizes deviation of closed-loop performance from a reference model.
	
	An important novel contribution of this paper is the way we parameterize the scheduling function of the qLPV model:
	we use a neural network with softmax output layer, so that the outputs of the network belong to the unit simplex by construction. This implies that the model matrices of the system being identified serve directly as vertices of a linear system with multiplicative uncertainty that bounds the nonlinear system. Then, an RCI set synthesized for the uncertain linear system is also RCI for the qLPV model, at the price of some possible conservativeness. To make the identification problem computationally tractable, we parameterize the RCI set as a polytope with fixed normal directions, and enforce configuration-constraints on the offsets~\cite{Villaneuva2024}. This enables us to express invariance conditions as being jointly linear in the parameters of the model and of the set, that we exploit to formulate an {\color{black} optimization problem} whose value serves as our control-oriented regularization.
	
	The paper is organized as follows. In Section~\ref{sec:definition}, we first state the learning problem we want to tackle. Then, we introduce the notions of uncertainty derived from model mismatch and set invariance, using which we formulate the conceptual SysID problem with control-oriented regularization. In Section~\ref{sec:parameterization}, we introduce the qLPV and RCI parameterizations, using which we formulate a computer-implementable version of the problem synthesized in Section~\ref{sec:definition}. 
	In Section~\ref{sec:sysID}, we present the concurrent-identification problem, along with schemes to identify the system orders and RCI set template, as well as an initial point for the problem. {\color{black}We also discuss the limitations of our approach in this section.}
	% In particular, in Sections~\ref{sec:init_identification} and~\ref{sec:init_qLPV_identification}, we elaborate upon the steps in this scheme, in which we also present some model reduction strategies. 
	%In Section~\ref{sec:concurrent_qLPV_identification}, we summarize the final identification procedure. 
	In Section~\ref{sec:numerical_examples}, we first present numerical examples to illustrate the effectiveness of the proposed qLPV parameterization with some benchmark examples. Furthermore, we test the model reduction schemes numerically. Finally, we solve the problem developed in Section~\ref{sec:parameterization} using data collected form a nonlinear mass-spring-damper system, and use this model and the associated optimal RCI set to synthesize an output tracking control scheme. 
	
	\subsection{Notation} 
	The symbols $\R$ and $\mathbb{N}$ denote the set of real and natural numbers, respectively. The set $\mathbb{I}_a^b:=\{a,\cdots,b\}$ is the set of indices between $a$ and $b$ with $a<b$. Given two sets $\mathcal{A},\mathcal{B} \subseteq \R^n$, their Minkowski sum is defined as $\mathcal{A}\oplus \mathcal{B}:=\{a+b:a\in \mathcal{A},b \in \mathcal{B}\}$. If $\mathcal{A}=\{a\}$ is singleton, with a slight abuse of notation we denote $\{a\}\oplus \mathcal{B}$ by $a \oplus \mathcal{B}$. The Cartesian product $\mathcal{A} \times \cdots \times \mathcal{A}$ taken $t$-times is denoted by $\mathcal{A}^t$.  Given a function $f:\R^{n} \to \R^m$, we denote the set-valued map $f(\mathcal{A}):=\{f(a) \in \R^m : a \in \mathcal{A}\}$.
	Given a vector $a \in \R^n$, $|a| \in \R^n$ denotes the element-wise absolute value vector. {\color{black}Given matrices $M_1,\cdots,M_N \in \R^{n \times m}$, we define their convex hull as the set of matrices
		$\mathrm{CH}\{M_i,i \in \mathbb{I}_1^N\}:=\{M\in \R^{n \times m}: M=\sum_{i=1}^N \lambda_i M_i,\ \sum_{i=1}^N \lambda_i=1,\ \lambda \geq 0\}$.}
	
	\section{Problem definition}
	\label{sec:definition}
	Assume that we have input and output data measurements available by exciting the nonlinear plant
	\begin{align}\label{eq:underlying_nonlinear}
		\mathbf{z}^+ = \mathbf{f}(\mathbf{z},u), &&
		y = \mathbf{g} (\mathbf{z}), 
	\end{align}
	where $u \in \R^{n_u}$ is the input, $y \in \R^{n_y}$ the output, $\mathbf{z} \in \R^{n_{\mathbf{z}}}$
	the state vector, and $\mathbf{z}^+$ denotes the successor state {\color{black}after a sample step}. We assume that the functions $\mathbf{f}$ and $\mathbf{g}$, as well as the state dimension $n_{\mathrm{z}}$, are unknown. 
	Our goal is to identify a control-oriented model of~\eqref{eq:underlying_nonlinear}. Precisely, we want to solve the following problem:
	\begin{problem}
		\label{prob:basic_problem}
		\textit{Given a dataset $\mathcal{D}:=\{(u_t,y_t),t \in \mathbb{I}_0^{N-1}\}$ 
			of input-output measurements {\color{black}collected} from~\eqref{eq:underlying_nonlinear}, identify a model {\color{black}capable of accurately predicting the plant’s behavior. Moreover, the identified model should be suitable for the synthesis of a feedback controller that regulates the plant output $y$ within a prescribed set $\mathbb{Y}$ under input constraints $u \in \mathbb{U}$.}} \\
	\end{problem}
	
	Problem~\eqref{prob:basic_problem} describes the entire learning-based control design pipeline, i.e., system identification and constrained controller synthesis. This problem is usually tackled in a cascaded way: first identify a model that best fits the available training/test data, then synthesize a controller. However, it may happen that, after identifying the model, synthesizing a controller that robustly satisfy output constraints under limited
	actuation is an impossible task, due to having separated the identification and robust control concerns.  The remaining part of this section is devoted to formulate a conceptual problem of {\it system identification with controller synthesis guarantees}.
	
	{\color{black}Before proceeding further, some clarifications regarding the structural assumptions on the underlying plant are due. Firstly, as we rely strictly on input-output data in $\mathcal{D}$, the identified model will inherently reside within the observable subspace of~\eqref{eq:underlying_nonlinear}. Secondly, as we will see later, our approach to providing controller synthesis guarantees relies on the feasibility of a control invariant set satisfying the constraints. Consequently, the issue of controllability is handled implicitly: If the observable subspace is fully controllable, our procedure seeks to find a model that admits a feasible control invariant set, hence solving Problem \ref{prob:basic_problem}. Conversely, if the observable subspace contains uncontrollable modes, our procedure yields a feasible solution only if these modes are stabilizable within the constraints. Otherwise, it returns infeasibility, implying that Problem \ref{prob:basic_problem} does not have a feasible solution.}
	
	Let us first consider the following model
	\begin{align}
		\label{eq:conceptual_model}
		x^+ = f(x,u), && \hat{y} = g(x)
	\end{align}
	of~\eqref{eq:underlying_nonlinear}, where $x \in \R^{n_x}$ is the state of the model. The learning problem is usually addressed by solving the optimization problem
	\begin{align}
		\label{eq:pure_sysID}
		&\min_{f,g,x_0} \ \ \frac{1}{N}\sum_{t=0}^{N-1} \|y_t - g(x_t)\|_2^2 \\
		& \ \ \text{s.t.} \ \ \ x_{t+1}=f(x_t,u_t), \ t \in \mathbb{I}_0^{N-1}, \nonumber
	\end{align}
	where $(u_t,y_t)\in\mathcal{D}$, possibly under $\ell_2$-regularization on the coefficients parameterizing $f$ and $g$.
	While the solution to Problem~\eqref{eq:pure_sysID} might render the subsequent model-based controller synthesis feasible, this is not guaranteed. To ameliorate this, we introduce a \textit{control-oriented regularization} into Problem~\eqref{eq:pure_sysID}.
	This regularization is based on the observation that the output of~\eqref{eq:conceptual_model} may not exactly match the output of~\eqref{eq:underlying_nonlinear}, such that a controller synthesized using~\eqref{eq:conceptual_model} may lack \textit{robustness} against the model mismatch, potentially leading to constraint violation. We address this using the state observer model
	\begin{align}
		\label{eq:conceptual_observer}
		z^+ = f(z,u)+\zeta(z,u,w), && \hat{y}=g(z),
	\end{align}
	where $\zeta$ is the observer function, and $w \in \R^{n_y}$ is the output discrepancy defined as $w:=y-g(z).$
	{\color{black}
		We will now characterize an uncertain model using \eqref{eq:conceptual_observer} that captures the
		modeling error in~\eqref{eq:conceptual_model} based on the following assumption.
		\begin{assumption}
			\label{ass:underlying_plant}
			The behaviour of system~\eqref{eq:underlying_nonlinear} can equivalently be described by the model
			\begin{align}
				\label{eq:underlying_fake}
				\hat{x}^+=f(\hat{x},u), && y=g(\hat{x})+v,
			\end{align}
			where $v \in \mathbb{V} \subseteq \R^{n_y}$ is an unknown measurement noise. 
		\end{assumption}
		
		We emphasize that the set $\mathbb{V}$ is not an intrinsic property of the system~\eqref{eq:underlying_nonlinear} known a priori. Rather, it characterizes the quality of the identified model $(f,g)$ with respect to the true system. A more accurate model results in a ``smaller'' set $\mathbb{V}$. In practice, $\mathbb{V}$ is estimated \textit{a posteriori} in a data-driven manner: once $(f,g)$ are fixed by solving~\eqref{eq:pure_sysID}, $\mathbb{V}$ is constructed as the set containing the output residuals $y_t - g(\hat{x}_t)$ observed during simulation. Furthermore, while we refer to $v$ as \textit{measurement noise} according to standard nomenclature, we emphasize that it must be accounted for when enforcing $y \in \mathbb{Y}$. We use the observer in \eqref{eq:conceptual_observer} to estimate the state of \eqref{eq:underlying_fake}. By defining the observer error as $e:=\hat{x}-z$, the error dynamics are given by
		\begin{align}
			\label{eq:fake_error_dynamics}
			e^+ = f(\hat{x},u)-f(z,u)-\zeta(z,u,g(\hat{x})+v-g(z)).
		\end{align}
		\begin{assumption}
			\label{ass:stable_observer}
			The observer function $\zeta$ renders the error dynamics in \eqref{eq:fake_error_dynamics} robustly stable, i.e., by denoting the input and measurement noise sequences as $\mathbf{u}_t:=(u_0,\cdots,u_{t-1})$ and $\mathbf{v}_t:=(v_0,\cdots,v_{t-1})$, respectively, there exists a robust positive invariant set $\mathcal{E} \subseteq \R^{n_x}$ such that 
			\begin{align}
				\label{eq:PI_error}
				e_0 \in \mathcal{E} \ \Rightarrow \ e_t \in \mathcal{E}, && \forall t \in \mathbb{N},
			\end{align}
			for all input and noise sequences $\mathbf{u}_t \in \mathbb{U}^t$ and $\mathbf{v}_t \in \mathbb{V}^t$.
		\end{assumption}
		
		Assumption \ref{ass:stable_observer} states that the error remains bounded in a set $\mathcal{E}$ for all time irrespective of the input and noise sequences. Note that under this assumption, the input sequence is not parameterized by a controller; instead, invariance is required to hold for all admissible input sequences. While this assumption might seem strong, it is routinely employed in observer-based control design, see, e.g.,~\cite{Mayne2006}. For example, suppose that System~\eqref{eq:underlying_fake} is described by the linear time-invariant (LTI) dynamics $\hat{x}^+ = A\hat{x}+Bu$ and $y=C\hat{x}+v$
		for some measurement noise $v \in \mathbb{V}$, and the observer model in~\eqref{eq:conceptual_observer} is described by $z^+ = Az+Bu+K(y-Cz)$ with observer gain $K$.
		Then, the minimal robust positive invariant set
		\begin{align*}
			\mathcal{E}=\bigoplus_{t=0}^{\infty}(A-KC)^t (-K \mathbb{V}),
		\end{align*}
		satisfies~\eqref{eq:PI_error}, and is bounded if $\rho(A-KC)<1$~\cite{Kolmanovsky1998}. Similar results for Lipschitz continuous nonlinear systems were presented in \cite{Phanomchoeng2010}.
		
		Towards deriving the uncertain model, we note that since Assumptions \ref{ass:underlying_plant} and \ref{ass:stable_observer} ensure the existence of sets $\mathbb{V}$ and $\mathcal{E}$ respectively, there always exists some disturbance set $\mathbb{W} \subseteq \R^{n_y}$ that encapsulates the measurement noise and observer error by satisfying the inclusion 
		\begin{align}
			\label{eq:W_requirement}
			\{g(\hat{x})+v-g(z) \ : \ \hat{x}-z \in \mathcal{E}, \ v \in \mathbb{V}\} \subseteq \mathbb{W}.
		\end{align}
		For the LTI example presented above, the set $\mathbb{W}=C\mathcal{E} \oplus \mathbb{V}$ satisfies \eqref{eq:W_requirement}.
		In the sequel, we treat $\mathbb{W}$ as a design parameter that must satisfy the following assumption.
		\begin{assumption}
			\label{ass:W_assumption}
			The disturbance set $\mathbb{W} \subseteq \R^{n_y}$ satisfies the inclusion in \eqref{eq:W_requirement}.
		\end{assumption}
		\iffalse	
		we make the following assumption.
		\begin{assumption}
			\label{ass:W_assumption}
			There exists a disturbance set $\mathbb{W} \subseteq \R^{n_y}$ satisfying the inclusion
			\begin{align}
				\label{eq:W_requirement}
				\{g(\hat{x})+v-g(z) \ : \ \hat{x}-z \in \mathcal{E}, \ v \in \mathbb{V}\} \subseteq \mathbb{W},
			\end{align}
			which encapsulates the observer error and measurement noise.
		\end{assumption}
		
		Under Assumptions \ref{ass:underlying_plant} and \ref{ass:stable_observer}, there always exists some disturbance set $\mathbb{W}$ that satisfies \eqref{eq:W_requirement} since the assumptions ensure the existence of sets $\mathcal{E}$ and $\mathbb{V}$.
		For the LTI example presented above, the set $\mathbb{W}=C\mathcal{E} \oplus \mathbb{V}$ satisfies \eqref{eq:W_requirement}. 
		\fi
		
		For a given input sequence $\mathbf{u}_t \in \mathbb{U}^t$ and disturbance sequence $\mathbf{w}_t \in \mathbb{W}^t$, let us denote by $\hat{y}(z_0,\mathbf{u}_t,\mathbf{w}_t)$ the output of \eqref{eq:conceptual_observer} at time $t$ from initial state $z_0 \in \R^{n_x}$. Then, the reachable set of outputs of \eqref{eq:conceptual_observer} at time $t$ for all disturbance sequences and a fixed input sequence is given by
		\begin{align*}
			\hat{Y}_t(z_0,\mathbf{u}_t):=\{y\in\mathbb{R}^{n_y}:\ \exists \mathbf{w}_t \in \mathbb{W}^t
			:\ y=\hat{y}(z_0,\mathbf{u}_t,\mathbf{w}_t)\}
		\end{align*}
		\begin{proposition}
			\label{prop:basic_observer}
			Suppose Assumptions \ref{ass:underlying_plant}, \ref{ass:stable_observer} and \ref{ass:W_assumption} hold. From a given initial state $\hat{x}_0 \in \R^{n_x}$, denote the output of \eqref{eq:underlying_fake} for any given input sequence $\mathbf{u}_t \in \mathbb{U}^t$ and noise sequence  $\mathbf{v}_t \in \mathbb{V}^t$ as  $y(\hat{x}_0,\mathbf{u}_t,\mathbf{v}_t)$. Then, the inclusion
			\begin{align}
				\label{eq:robustness_guarantee}
				y(\hat{x}_0,\mathbf{u}_t,\mathbf{v}_t) \in \hat{Y}_t(z_0,\mathbf{u}_t), && \forall t \in \mathbb{N}
			\end{align}
			holds if initial states $\hat{x}_0$ and $z_0$ satisfy $\hat{x}_0-z_0 \in \mathcal{E}$.
		\end{proposition}
		\begin{proof}
			Under Assumptions  \ref{ass:underlying_plant} and \ref{ass:stable_observer}, there exists some set $\mathbb{W} \subseteq \R^{n_y}$ satisfying Assumption \ref{ass:W_assumption}.
			The proof concludes if for any noise sequence $\mathbf{v}_t \in \mathbb{V}^t$, there exists some disturbance sequence $\mathbf{w}_t \in \mathbb{W}^t$ such that $y(\hat{x}_0,\mathbf{u}_t,\mathbf{v}_t)=\hat{y}(z_0,\mathbf{u}_t,\mathbf{w}_t)$ holds if $\hat{x}_0 - z_0 \in \mathcal{E}$. Clearly, this is satisfied by the sequence 
			\begin{align*}
				\tilde{\mathbf{w}}_t=(y(\hat{x}_0,\mathbf{u}_1,\mathbf{v}_1)&-g(z(z_0,\mathbf{u}_1,\mathbf{v}_1)),  \\  & \cdots, y(\hat{x}_0,\mathbf{u}_t,\mathbf{v}_t)-g(z(z_0,\mathbf{u}_t,\mathbf{v}_t))).
			\end{align*} 
			It remains to show that $\tilde{\mathbf{w}}_t \in \mathbb{W}^t$. Observe that from~\eqref{eq:PI_error}, the sequence $\tilde{\mathbf{w}}_t$ ensures that $\hat{x}(\hat{x}_0,\mathbf{u}_t) - z(z_0,\mathbf{u}_t,\mathbf{v}_t) \in \mathcal{E}$, and the condition on $\mathbb{W}$ in~\eqref{eq:W_requirement} that $\tilde{\mathbf{w}}_t \in \mathbb{W}^t$.
		\end{proof}
		
		Hence, we can define the uncertain dynamical system
		\begin{align}
			\label{eq:robust_model}
			z^+ \in f(z,u) \oplus \zeta(z,u,\mathbb{W}), && y \in g(z) \oplus \mathbb{W}.
		\end{align}
		This uncertain model serves as a set-valued surrogate of the true system. Its utility is grounded in the guarantee provided by Proposition~\ref{prop:basic_observer}: provided that \eqref{eq:robust_model} is initialized such that its state is consistent with that of \eqref{eq:underlying_fake}, i.e., $\hat{x}_0-z_0 \in \mathcal{E}$, we can simulate~\eqref{eq:robust_model} in a set-valued manner to encompass all possible future behaviors of \eqref{eq:underlying_fake} since \eqref{eq:robustness_guarantee} holds. Consequently, if a controller is designed to keep the reachable output set $\hat{Y}_t(z_0,\mathbf{u}_t)$ of \eqref{eq:robust_model} within the safety constraints, the true system~\eqref{eq:underlying_nonlinear} is guaranteed to satisfy them. The following result can be used to guarantee the existence of a feasible control law which ensures $\hat{Y}_t(z_0,\mathbf{u}_t) \subseteq \mathbb{Y}$ for all $t \in \mathbb{N}$.
	}
	{\color{black}
		\begin{proposition}
			\label{prop:RCI_exists}
			Suppose Assumptions \ref{ass:underlying_plant}, \ref{ass:stable_observer} and \ref{ass:W_assumption} hold. Suppose further that there exists a set $X \subseteq \R^{n_x}$ and a control law $\mu(\cdot):X \to \mathbb{U}$ that satisfy the inclusion
			\begin{subequations}
				\label{eq:RCI_fundamentals}
				\begin{align}
					\label{eq:RCI_dynamics}
					f(z,\mu(z)) \oplus \zeta(z,\mu(z),\mathbb{W}) &\subseteq X, && \forall z \in X.
				\end{align}
			\end{subequations}
			Then, if the initial state of the observer \eqref{eq:conceptual_observer} satisfies $z_0 \in X$, initial state of the plant in \eqref{eq:underlying_fake} satisfies $\hat{x}_0-z_0 \in \mathcal{E}$, and the set $X$ satisfies the output inclusion
			\addtocounter{equation}{-1} 
			\begin{subequations}
				% 2. Push the sub-equation letter forward from 'a' to 'b'
				\addtocounter{equation}{1} 
				\begin{align}
					\label{eq:output_inclusion}
					g(z) \oplus \mathbb{W} &\subseteq \mathbb{Y}, && \forall z \in X,
				\end{align}
			\end{subequations}
			the control input $u_t=\mu(z_t)$, where $z_t$ is the state of \eqref{eq:conceptual_observer} at time $t$ from initial state $z_0$ and any arbitrary disturbance sequence $\mathbf{w}_t \in \mathbb{W}^t$ ensures that the plant output satisfies $y_t \in \mathbb{Y}$ for all $t \in \mathbb{N}$, thus guaranteeing constraint satisfaction.
		\end{proposition}
		\begin{proof}
			Under Assumptions \ref{ass:underlying_plant}, \ref{ass:stable_observer} and \ref{ass:W_assumption}, inclusion \eqref{eq:robustness_guarantee} holds for any $\mathbf{u}_t \in \mathbb{U}^t$ and $\mathbf{v}_t \in \mathbb{V}^t$ if $\hat{x}_0-z_0 \in \mathcal{E}$ as per Proposition \ref{prop:basic_observer}. Then, given inclusion \eqref{eq:RCI_dynamics}, if \eqref{eq:conceptual_observer} is initialized with $z_0 \in X$, then we have $z_1 \in X$ with $u_0 = \mu(z_0)$ for all disturbances $w_0 \in \mathbb{W}$. From \eqref{eq:output_inclusion} and \eqref{eq:robustness_guarantee}, this implies that the plant satisfies $y_1 \in \mathbb{Y}$ for all $v_1 \in \mathbb{V}$. Assume $z_t \in X$, and apply $u_t = \mu(z_t)$. Then we have $z_{t+1} \in X$ for any $w_t \in \mathbb{W}$ as per condition~\eqref{eq:RCI_dynamics}. Consequently, by \eqref{eq:output_inclusion} and Proposition~\ref{prop:basic_observer}, the output $y_{t+1} \in \mathbb{Y}$. Thus, $y_t \in \mathbb{Y}$ for all $t \in \mathbb{N}$.
		\end{proof}
	}
	The conditions in~\eqref{eq:RCI_fundamentals} enforce $X$ to be a robust control invariant (RCI) set for the uncertain system in \eqref{eq:robust_model}, with invariance induced by the feedback law $u=\mu(z)$. Since this set $X$ satisfies the output constraints, it guarantees the existence of a feasible controller for the plant, thus solving Problem \eqref{prob:basic_problem}. While any set $X$ satisfying~\eqref{eq:RCI_fundamentals} guarantees safety, we aim to construct one that is minimally conservative to maximize the controller's operating range. {\color{black} To this end, we introduce the function $\mathbf{r}(f,\zeta,g)$ in the sequel which satisfies
		\begin{subequations}
			\label{eq:r_reqmts}
			\begin{align}
				&\mathbf{r}(f,\zeta,g)< \infty \Rightarrow \exists X \text{ satisfying \eqref{eq:RCI_fundamentals}} \label{eq:r_reqmt_1}\\
				&\textit{Small } \mathbf{r}(f,\zeta,g)< \infty \ \Rightarrow \ \exists X \text{ satisfying \eqref{eq:RCI_fundamentals}}  \label{eq:r_reqmt_2} \\
				& \hspace{110pt} \text{with minimal conservativeness.} \nonumber
			\end{align}
		\end{subequations}
		The function $\mathbf{r}(f,\zeta,g)$ is designed to ensure the existence of an RCI set $X$, with small values indicating that the set $X$ is minimally conservative. While \eqref{eq:r_reqmts} characterizes qualitative requirements on the regularization function, we quantify them in the sequel.
		Using this function, we modify Problem \eqref{eq:pure_sysID} as
		\begin{align}
			\hspace{-2pt} &\min_{f,\zeta,g,x_0} \ \ \frac{1}{N}\sum_{t=0}^{N-1} \|y_t - g(x_t)\|_2^2 + \tau \mathbf{r}(f,\zeta,g) \nonumber \\
			\hspace{-2pt} & \ \ \ \  \text{s.t.} \ \ \ \ x_{t+1}=f(x_t,u_t), \ t \in \mathbb{I}_0^{N-1}.     \label{eq:regularized_sysID}
		\end{align}
		where $\tau \geq 0$ is the regularization parameter. We label $\mathbf{r}(f,\zeta,g)$ as \textit{control-oriented regularization}. As per \eqref{eq:r_reqmts}, any triplet $(f,\zeta,g)$ that admits a finite value can be used to design a controller for the underlying plant, with solutions corresponding to a small value admitting controllers with minimal conservativeness. In the sequel, we present a parameterization of the objects involved in the formulation of Problem \eqref{eq:regularized_sysID}, using which we derive a tractable formulation of the optimization problem.
		
		\subsubsection*{Connections to existing literature}
		Our contribution aims to identify open-loop plant models for set-based robust control. Traditional approaches for jointly characterizing a model and its associated uncertainty, such as \cite{Milanese2004, Lauricella2020, Cerone2008}, typically operate under conditions similar to Assumption \ref{ass:underlying_plant} with a known set $\mathbb{V}$. While these methods yield a direct characterization of the uncertain (potentially nonlinear) model, they do not guarantee the existence of an RCI set, as the output constraint set $\mathbb{Y}$ is not incorporated into the identification procedure. 
		In learning-based control frameworks like \cite{Terzi2019}, the set $\mathbb{V}$ is estimated along with an uncertain linear model for use in robust control synthesis. However, because the constraints $\mathbb{Y}$ are decoupled from the identification phase, the resulting control synthesis problem may still be infeasible. Recent contributions, such as \cite{Mulagaleti2022, Chen2022}, address this by identifying a model and a set $\mathbb{V}$ while explicitly enforcing that the model admits an RCI set within $\mathbb{Y}$, thus coupling system identification with RCI set synthesis. In these works, a parameterization of the RCI set is selected a priori, with \cite{Mulagaleti2022} deriving rigorous bounds to satisfy Assumption \ref{ass:underlying_plant}. Nevertheless, these approaches are currently limited to identifying uncertain linear systems from datasets containing full state measurements. Our approach extends these ideas to the identification of uncertain nonlinear systems using only input-output data, while guaranteeing the existence of a feasible RCI set.
		A related line of research involves safe Reinforcement Learning (RL) with robust MPC parameterizations~\cite{Zanon2021}. These methods also tackle the estimation of $\mathbb{V}$ online, and parameterize $X$ as the feasible region of the MPC controller. Finally, behavioral system~\cite{Willems2005} approaches, such as \cite{Berberich2021}, have gained attention for synthesizing robust controllers directly from noisy input-output trajectories. While originally designed for linear systems, recent extensions aim to generalize these capabilities to nonlinear models \cite{Martin2023,Verhoek2023_DD}. A detailed comparison of our approach with these data-driven methods is a subject for future work.
	}
	
	\section{Model and set parameterization}
	In this section, we present a model and RCI set parameterization, using which we formulate Problem~\eqref{eq:regularized_sysID}. We parameterize the model as a qLPV system. Such models describe state evolution using a time-varying linear map, with the linear maps scheduled using an explicitly characterized function~\cite{Toth2010Book}. They have frequently been observed to provide a very good balance between prediction accuracy and ease of robust controller design. In this work, we exploit a particular choice of scheduling function parameterization that enables a straightforward derivation of a linear system with multiplicative uncertainty that bounds the qLPV system. Then, we parameterize the RCI set as a polytope with fixed normal vectors, in which we induce robustness for the uncertain linear system using linear inequalities.  Finally, we present a formulation of the control-oriented regularization function $\mathbf{r}$ for our parameterization.
	\label{sec:parameterization}
	\subsection{Dynamical system}
	We parameterize model~\eqref{eq:conceptual_model} as the qLPV system
	\begin{align}
		x^+ = A(p(x,u))x + B(p(x,u))u, &&
		\hat{y} = Cx, \label{eq:LPV_model}
	\end{align}
	and the state observer in \eqref{eq:conceptual_observer} with ${\zeta(z,u,w)=K(p(z,u))w}$, where the
	matrix-valued functions $(A(p),B(p),K(p))$ depend linearly on the scheduling vector $p \in \R^{n_p}$ as
	\begin{align}
		\label{eq:linear_parameterization}
		(A(p), B(p),K(p)) := \sum_{i=1}^{n_p}p_i(A_i, B_i,K_i).
	\end{align}
	Further, we model each component of the scheduling function $\scalemath{0.95}{p : \R^{n_x + n_u} \to \R^{n_p}}$ as
	\begin{align}
		p_i(x,u) = \begin{cases} &\cfrac{e^{\mathcal{N}_i(x,u)}}{1+\sum_{j=1}^{n_p-1} e^{\mathcal{N}_j(x,u)}},\  i \in \mathbb{I}_1^{n_p-1},  \\
			&\cfrac{1}{1+\sum_{j=1}^{n_p-1}e^{\mathcal{N}_j(x,u)}}, \ i=n_p.
		\end{cases} \label{eq:ss-p}
	\end{align}
	where each $\mathcal{N}_i:\R^{n_x + n_u} \to \R$  is a feedforward neural network (FNN) with $n_h$ hidden layers. The parameterization~\eqref{eq:ss-p} forces $p$ to belong to the simplex $\mathcal{P}$ defined as
	\begin{align}
		\label{eq:simplex_definition}
		\mathcal{P}:=\left\{ p : \sum_{i=1}^{n_p} p_i= 1, \ 0 \leq p \leq 1\right\}.
	\end{align}
	Note that~\eqref{eq:ss-p} can be interpreted as the output of a classifier
	for predicting a multi-category target of dimension $n_p$, where $p_i$
	is the probability of the target being in category $i$, i.e., $p(x,u)$ is a feedforward neural network with softmax output. {\color{black}
		Consequently, the parameterization in \eqref{eq:ss-p} makes \eqref{eq:LPV_model} a Mixture-of-Experts (MoE) model. MoE architectures are known universal approximators for continuous functions on compact sets \cite{Nguyen2016}. While the structure of \eqref{eq:LPV_model} assumes the origin to be an equilibrium point, affine dynamics can be readily accommodated by augmenting the state vector as $\tilde{x} = [x^\top \ 1]^\top$. Therefore, provided that the plant dynamics in~\eqref{eq:underlying_nonlinear} are continuous, the proposed qLPV parameterization is sufficiently expressive to satisfy Assumption \ref{ass:underlying_plant}.
	}
	{\color{black}
		\begin{remark}
			It is important to distinguish the proposed identification approach from the ``local approach'' to LPV identification, in which independent LTI models are estimated at fixed operating points in arbitrary statespaces, requiring coherence transformations to unify them before interpolation \cite{Zhang2020}. In contrast, we employ a global identification approach, e.g., \cite{Verhoek2022}, using a self-scheduled parameter $p(x,u)$ as in \eqref{eq:ss-p}, and identify the global coefficient matrices $(A_i,B_i,K_i)$ along with the scheduling function directly. This construction ensures that the state $x$ is defined in a unified basis, rendering coherence transformations unnecessary. Furthermore, identification is performed directly using an input-output dataset without requiring a scheduling variable dataset.
		\end{remark}
	}
	
	\subsection{Robust control invariant set}
	We parameterize the RCI set $X \subseteq \R^{n_x}$ as the polytope
	\begin{align*}
		X \leftarrow \mathbb{X}(q):=\{x \in \R^{n_x} : Fx \leq q\},
	\end{align*}
	where $F \in \R^{\fe \times n_x}$ is a matrix that we fix a priori, and $q \in \R^{\fe}$ is the variable that characterizes the set. Over $q$, we enforce \textit{configuration-constraints}~\cite{Villaneuva2024}, which are conic constraints of the form $\mathcal{C}:=\{q : Eq \leq 0\}$. These constraints dictate that
	\begin{align}
		\label{eq:config_con}
		q \in \mathcal{C} \ \Rightarrow \ \mathbb{X}(q)=\mathrm{CH}\{V_l q, l \in \mathbb{I}_1^{\ve}\},
	\end{align}
	where $\{V_l \in \R^{n_x \times \fe}, l \in \mathbb{I}_1^{\ve}\}$ are the vertex maps given a priori. Then, we parameterize the disturbance set as
	\begin{align}
		\label{eq:parameterization_W}
		\mathbb{W} \leftarrow \mathbb{W}_{\kappa}(c_{\mathrm{w}},\epsilon_{\mathrm{w}}):=\{w \in \R^{n_y} : |w - c_{\mathrm{w}}| \leq \kappa \epsilon_{\mathrm{w}}\},
	\end{align}
	where $c_{\mathrm{w}},\epsilon_{\mathrm{w}} \in \R^{n_y}$ are the parameters, and $\kappa>1$ is a user-specified parameter to account for finite data. We present a characterization of the disturbance set parameters in the sequel.
	Firstly, we enforce the set $\mathbb{X}(q)$ to be RCI for the uncertain system in~\eqref{eq:robust_model}, written for the qLPV parameterization as
	\begin{align}
		z^+ &\in A(p(z,u))z+B(p(z,u))u \oplus K(p(z,u)) \mathbb{W}_{\kappa}(c_{\mathrm{w}},\epsilon_{\mathrm{w}}), \nonumber \\
		y &\in Cz \oplus \mathbb{W}_{\kappa}(c_{\mathrm{w}},\epsilon_{\mathrm{w}}) \label{eq:uncertain_qLPV}
	\end{align}
	{\color{black}
		The following result is key to our developments.
		\begin{proposition}
			\label{prop:RCI_transfer}
			Let the convex hull of the matrices $(A_i,B_i,K_i)$ be denoted by $\Delta := \mathrm{CH}\{(A_i,B_i,K_i),i \in \mathbb{I}_1^{n_p}\}$. An RCI set $\mathbb{X}(q)$ for the uncertain linear system
			\begin{align}
				\label{eq:multiplicative_RCI}
				z^+ &\in Az+Bu\oplus  K\mathbb{W}_{\kappa}(c_{\mathrm{w}},\epsilon_{\mathrm{w}}),  \nonumber \\
				y &\in Cx \oplus \mathbb{W}_{\kappa}(c_{\mathrm{w}},\epsilon_{\mathrm{w}})
			\end{align}
			with multiplicative uncertainty $(A,B,K) \in \Delta$ is RCI also for the uncertain nonlinear system in \eqref{eq:uncertain_qLPV}.
		\end{proposition}
		\begin{proof}
			Take any $z \in \mathbb{X}(q)$. Since $\mathbb{X}(q)$ is RCI for the linear system in \eqref{eq:multiplicative_RCI}, there exists some input $u \in \mathbb{U}$ such that $Az+Bu+Kw \in \mathbb{X}(q)$ for all $w \in \mathbb{W}_{\kappa}(c_{\mathrm{w}},\epsilon_{\mathrm{w}})$ and $(A,B,K) \in \Delta$. From the parameterizations in \eqref{eq:linear_parameterization}, \eqref{eq:ss-p} and \eqref{eq:simplex_definition}, we have that $(A(p(z,u)),B(p(z,u)),K(p(z,u))) \in \Delta$, which implies the subsequent state of \eqref{eq:uncertain_qLPV} also satisfies
			\begin{align*}
				A(p(z,u))z+B(p(z,u))u+K(p(z,u))w \in \mathbb{X}(q)
			\end{align*}
			for all $w \in \mathbb{W}_{\kappa}(c_{\mathrm{w}},\epsilon_{\mathrm{w}})$. Hence, $\mathbb{X}(q)$ is RCI also for \eqref{eq:uncertain_qLPV}.
		\end{proof}
	}
	
	Exploiting Proposition \ref{prop:RCI_transfer}, we now characterize conditions on $q$ such that $\mathbb{X}(q)$ is an RCI set for~\eqref{eq:multiplicative_RCI}. In this result, we assume that the constraint sets $\mathbb{U}$ and $\mathbb{Y}$ are defined as
	\begin{align*}
		\mathbb{U}:=\{u \in \R^{n_u} : H^u u \leq h^u\}, \  \mathbb{Y}:=\{y \in \R^{n_y} : H^y y \leq h^y\}.
	\end{align*}
	\begin{proposition}
		\label{prop:RCI_multiplicative}
		A set $\mathbb{X}(q)$ is an RCI set for system~\eqref{eq:multiplicative_RCI} with respect to constraints $y \in \mathbb{Y}$ and $u \in \mathbb{U}$ if there exist control inputs $\mathrm{u}:=(\mathrm{u}_1,\cdots,\mathrm{u}_{\ve})$ such that the inequalities
		\begin{subequations}
			\label{eq:RCI_CC}
			\begin{align}
				F(A_iV_l q + B_i \mathrm{u}_l+K_i c_{\mathrm{w}}) + \kappa|FK_i|\epsilon_{\mathrm{w}} &\leq q, \label{eq:RCI_CC_1}\\
				H^y (CV_l q + c_{\mathrm{w}})+ \kappa|H^y| \epsilon_{\mathrm{w}} &\leq h^y, \label{eq:RCI_CC_2} \\
				H^u \mathrm{u}_l \leq h^u, \ \  E q &\leq 0, \label{eq:RCI_CC_4}
			\end{align}
		\end{subequations}
		are verified for all $i \in \mathbb{I}_1^{n_p}$ and $l \in \mathbb{I}_1^{\ve}$.
	\end{proposition}
	\begin{proof}
		The proof follows from definition of RCI sets and ~\cite[Corollary 4]{Villaneuva2024} which exploits \eqref{eq:config_con}, where $\mathrm{u}$ denotes the vertex control inputs.
	\end{proof}
	
	For fixed system matrices $(A_i,B_i,K_i,C)$ and disturbance set parameters $(\mathrm{c}_{\mathrm{w}},\epsilon_{\mathrm{w}},\kappa)$, the conditions in \eqref{eq:RCI_CC} are linear in $(q,\mathrm{u})$. We exploit this property in the development of the regularization function $\mathbf{r}$ in the sequel.
	
	\subsection{Disturbance set characterization}
	We recall from Proposition \ref{prop:RCI_exists} that the RCI set $\mathbb{X}(q)$ of \eqref{eq:uncertain_qLPV} can be used to guarantee the existence of a controller for the underlying plant if the disturbance set $\mathbb{W}_{\kappa}(\mathrm{c}_{\mathrm{w}},\epsilon_{\mathrm{w}})$ satisfies Assumption \ref{ass:W_assumption}. We now present a procedure to estimate the disturbance set parameters $(\mathrm{c}_{\mathrm{w}},\epsilon_{\mathrm{w}},\kappa)$ such that this assumption is satisfied. Assuming access to an additional dataset
	$\mathcal{D}^{\mathrm{w}}:=\{(y_t,u_t), t \in \mathbb{I}_0^{N_{\mathrm{w}-1}}\}$ of input-output measurements from the plant, we simulate the closed-loop observer
	\begin{align}
		\label{eq:closed_loop_observer}
		\hspace{-5pt} \scalemath{0.98}{z^+=A(p(z,u))z+B(p(z,u))u+K(p(z,u))(y-Cz)},
	\end{align}
	(equivalent to \eqref{eq:conceptual_observer}) using inputs and outputs from $\mathcal{D}^{\mathrm{w}}$, from $z_0=0$. Then, building a sampled set of the output discrepancy as $\mathcal{W}:=\{y_t-Cz_t, t \in \mathbb{I}_0^{N_{\mathrm{w}-1}}\}$,
	we can define the prediction error bounds taken elementwise as $\overline{w}:=\max_{w \in \mathcal{W}} w$ and $\underline{w}:=\min_{w \in \mathcal{W}} w.$
	Given the disturbance set parameterization in \eqref{eq:parameterization_W}, the inclusion $\mathcal{W} \subset \mathbb{W}_{\kappa}(c_{\mathrm{w}},\epsilon_{\mathrm{w}})$ then holds with
	\begin{align}
		\label{eq:W_parameters_unique}
		c_{\mathrm{w}}=0.5( \overline{w}+\underline{w}), && \epsilon_{\mathrm{w}}=0.5( \overline{w}-\underline{w})
	\end{align}
	for any $\kappa \geq 1$. Since $\mathcal{D}^{\mathrm{w}}$ is a finite dataset however, the disturbance set $\mathbb{W}_{\kappa}(c_{\mathrm{w}},\epsilon_{\mathrm{w}})$ might not satisfy Assumption \ref{ass:W_assumption} with $\kappa=1$. In the following result, we present a lower-bound on the inflation parameter $\kappa$ to ensure its satisfaction.
	{\color{black}
		\begin{proposition}
			\label{prop:kappa_bound}
			Suppose Assumptions \ref{ass:underlying_plant} and \ref{ass:stable_observer} are satisfied by the model structure in \eqref{eq:LPV_model} and observer function $K(p(z,u))w$ respectively, such that there exists an invariant set $\mathcal{E} \subseteq \R^{n_x}$ in which the observer error $e=\hat{x}-z$ is bounded. Suppose that $0 \in \mathcal{E}$, $\hat{x}_0 \in \mathcal{E}$, and the dataset $\mathcal{D}^{\mathrm{w}}$ is such that
			\begin{align}
				\label{eq:distance_assumption}
				\forall \delta \in C\mathcal{E}\oplus \mathbb{V}, \ \exists t \in \mathbb{I}_0^{N_{\mathrm{w}}-1} : |\delta-(y_t-Cz_t)| \leq \alpha,
			\end{align}
			for some $\alpha>0$, 
			where $z_t$ is the state of \eqref{eq:closed_loop_observer} simulated from $z_0 = 0$ using $\mathcal{D}^{\mathrm{w}}$. Then, if
			\begin{align}
				\label{eq:kappa_bound}
				\kappa \geq 1+\max_{i \in \mathbb{I}_1^{n_y}} \alpha_i/\epsilon_{\mathrm{w}_i},
			\end{align}
			the disturbance set $\mathbb{W}_{\kappa}(c_{\mathrm{w}},\epsilon_{\mathrm{w}})$ with $(c_{\mathrm{w}},\epsilon_{\mathrm{w}})$ estimated as in \eqref{eq:W_parameters_unique} satisfies Assumption \ref{ass:W_assumption}.
		\end{proposition}
		\begin{proof}
			Assumptions \ref{ass:underlying_plant} and \ref{ass:stable_observer}, along with $\hat{x}_0 \in \mathcal{E}$ and $z_0=0$ guarantee that $e_t=\hat{x}_t-z_t \in \mathcal{E}$ for all $t \in \mathbb{I}_0^{N_{\mathrm{w}-1}}$. For any plant and observer states $\hat{x},z \in \R^{n_x}$ satisfying the inclusion $C(\hat{x}-z) \in C\mathcal{E} \oplus \mathbb{V}$, we can then write
			\begin{align*}
				&|C(\hat{x}-z)-c_{\mathrm{w}}| \nonumber \\
				& \hspace{20pt}\leq |(y_t-Cz_t)-c_{\mathrm{w}}| + |C(\hat{x}-z)-(y_t-Cz_t)| \nonumber \\
				& \hspace{50pt}\leq \epsilon_{\mathrm{w}} + \alpha \leq \kappa \epsilon_{\mathrm{w}}, \nonumber
			\end{align*}
			where the first inequality follows from the triangle inequality, the second from \eqref{eq:W_parameters_unique} and \eqref{eq:distance_assumption} for some $t \in \mathbb{I}_0^{N_{\mathrm{w}}-1}$, and the third from~\eqref{eq:kappa_bound}. Hence, $C\mathcal{E} \oplus \mathbb{V} \subseteq \mathbb{W}_{\kappa}(c_{\mathrm{w}},\epsilon_{\mathrm{w}})$ holds, thus satisfying Assumption \ref{ass:W_assumption}.
		\end{proof}
		\begin{assumption}
			\label{ass:kappa_assumption}
			The inflation parameter $\kappa> 1$ is chosen to be large enough that \eqref{eq:kappa_bound} is verified.
		\end{assumption}
		
		Under Assumption \ref{ass:kappa_assumption}, the disturbance set $\mathbb{W}_{\kappa}(\mathrm{c}_{\mathrm{w}},\epsilon_{\mathrm{w}})$ is guaranteed to satisfy Assumption \ref{ass:W_assumption}, such that ensuring the existence of an RCI set $\mathbb{X}(q)$ ensures the existence of a controller for the underlying plant as per Proposition \ref{prop:RCI_exists}. In practice, a priori satisfaction of Assumption \ref{ass:kappa_assumption} is difficult to guarantee as $\alpha$ is generally unknown. While $\kappa \approx 1$ may suffice for a large dataset $\mathcal{D}^{\mathrm{w}}$ generated by persistently exciting inputs, the exact computation of $\kappa$ remains a fundamental challenge in data-driven control \cite{Mulagaleti2022, Zanon2021}. We remark that $\kappa$ acts as a \textit{robustness parameter}: larger values yield a larger disturbance set ensuring the RCI set is robust against significant model uncertainty, but potentially leading to infeasibility. While bootstrapping techniques utilizing multiple datasets $\mathcal{D}^{\mathrm{w}}$ could offer a way to estimate $\kappa$ more accurately, the development of such methods remains a subject for future research.

		\subsection{Control-oriented regularization function}
		We now present a formulation of the regularization function $\mathbf{r}$ satisfying the requirements in \eqref{eq:r_reqmts}. We denote by $\Theta$ the vector containing the parameters of the qLPV models in \eqref{eq:LPV_model} and \eqref{eq:uncertain_qLPV}, i.e., the matrices $(A_i,B_i,K_i)$, the output matrix $C$, and the weights of the FNNs parameterizing the scheduling function in \eqref{eq:ss-p}. Given $\Theta$, the disturbance set parameters $(\mathrm{c}_{\mathrm{w}}, \epsilon_{\mathrm{w}}, \kappa)$ are uniquely defined as per \eqref{eq:W_parameters_unique}. We recall that $\Theta$ parameterizes the linear inequalities in \eqref{eq:RCI_CC}. Accordingly, we define the set of admissible RCI parameters as
		\begin{align}
			\label{eq:RCI_set_parameters}
			\mathbb{Q}(\Theta):=\{q \in \R^{\fe} : \exists \mathrm{u} \in \R^{\ve n_u}: (q,\mathrm{u}) \ \text{satisfy} \  \eqref{eq:RCI_CC}\}.
		\end{align}
		As established in Propositions \ref{prop:RCI_transfer} and \ref{prop:RCI_multiplicative}, any vector $q \in \mathbb{Q}(\Theta)$ defines an RCI set $\mathbb{X}(q)$ for the uncertain nonlinear system \eqref{eq:uncertain_qLPV}. Hence, by Proposition \ref{prop:RCI_exists}, if $\mathbb{Q}(\Theta)$ is nonempty under the aforementioned assumptions, the identified model can be used to design a controller for the underlying plant.
		
		In designing the regularization function, we emphasize that the metric for the \textit{conservativeness} of an RCI set must align with the downstream control objectives. For instance, if the goal is robust stabilization to the origin, the function 
		\begin{align}
			\label{eq:mRCI}
			\mathbf{r}(\Theta):=\min_{q \in \mathbb{Q}(\Theta)} \|q\|_1  
		\end{align} 
		inspired by~\cite{Trodden2016} is appropriate, as it encourages a minimal RCI set. In this paper, our objective is to synthesize output tracking controllers. To this aim, we assume access to an output reference trajectory $\{y^{\mathrm{r}}_{k} \in \R^{n_y},k \in \mathbb{I}_0^{M}\}$, using which we define the regularization function as
		\begin{align}
			\label{eq:r_tracking}
			\	\mathbf{r}(\Theta) := & \min_{q, \mathbf{v}} \ \sum_{k=0}^M \|y_{k}^{\mathrm{r}}-Cz_{k}\|_2^2 \\
			& \ \text{s.t.} \ 
			\begin{cases}
				q \in \mathcal{Q}(\Theta),\  z_0 = 0, \ v_{k} \in \mathbb{U}, \ z_{k} \in \mathbb{X}(q),\\[1ex]
				z_{{k}+1} = A(p(z_{k},v_{k}))z_{k} + B(p(z_{k},v_{k}))v_{k}, \\[1ex]
				{k} \in \mathbb{I}_0^{M-1},
			\end{cases} \nonumber
		\end{align}
		and define $\mathbf{r}(\Theta)=\infty$ if Problem \eqref{eq:r_tracking} is infeasible.
		In \eqref{eq:r_tracking}, we compute an RCI set $\mathbb{X}(q)$ and a feasible input sequence $\mathbf{v}=(v_0,\cdots,v_{M-1})$ that makes the qLPV system track the reference while staying inside the RCI set. In Figure \ref{fig:rtheta_diag}, we present a schematic of the evaluation of $\mathbf{r}(\Theta)$. This formulation satisfies the requirements in \eqref{eq:r_reqmts}: if $\mathbb{Q}(\Theta)$ is empty, the value is $\infty$; if nonempty, a small value indicates that the open-loop system is capable of tight tracking, implying reduced conservativeness relevant to the control task.
		Problem \eqref{eq:r_tracking} can be adapted for multiple reference trajectories by considering a different open-loop sequence for each reference. Due to the qLPV dynamics, it is a nonlinear optimization problem. A convex alternative can be derived by replacing the qLPV dynamics with the uncertain LTI system in~\eqref{eq:multiplicative_RCI} and employing robust tube-based MPC techniques (e.g.,~\cite{Badalamenti2024}) to propagate the set-valued dynamics under additive disturbances. The development of such relaxations is left for future research. 
		\begin{figure}[t]
			\resizebox{\columnwidth}{!}{
				\begin{tikzpicture}[
					% Global styles for the diagram
					block/.style={
						draw, 
						rectangle, 
						thick, 
						align=center,
						minimum width=0.9cm,  % Adjust width here
						minimum height=0.4cm  % <--- ADJUST THIS NUMBER to change rectangle height
					},
					line/.style={
						->, 
						>=Stealth, % A standard, professional arrow tip
						thick
					},
					% Style for the connecting dot
					connector/.style={
						circle,
						fill,
						inner sep=1.5pt
					},
					node distance=1.75cm % Distance between the blocks
					]
					
					% --- Nodes (The Blocks) ---
					\node[block] (b1) {Eq. \eqref{eq:W_parameters_unique}}; 
					\node[block, right=of b1] (b2) {Eq. \eqref{eq:RCI_CC}}; 
					\node[block, right=of b2] (b3) {Eq. \eqref{eq:r_tracking}};
					
					% --- Arrows & Labels ---
					
					% 1. Input Arrow (Left -> Block 1)
					% Modified to include a coordinate 'split' for the branch
					\draw[line] ([xshift=-1.2cm]b1.west) -- node[midway, above] {$\Theta$} coordinate[pos=0.5] (split) (b1.west);
					
					% 2. Top Input Arrow (Top -> Block 1)
					\draw[line] ([yshift=0.6cm]b1.north) -- node[midway, right] {$\mathcal{D}^{\mathrm{w}}$} (b1.north);
					
					% 3. Arrow from Block 1 -> Middle Block
					\draw[line] (b1.east) -- node[midway, above] {$(\mathrm{c}_{\mathrm{w}}, \epsilon_{\mathrm{w}}, \kappa)$} (b2.west);
					
					% 4. Arrow from Middle Block -> Block 3
					\draw[line] (b2.east) -- node[midway, above] {$\mathcal{Q}(\Theta)$} (b3.west);
					
					% 5. Output Arrow (Block 3 -> Right)
					\draw[line] (b3.east) -- node[midway, above] {$\mathbf{r}(\Theta)$} +(1.cm, 0);
					
					% --- NEW: Theta Feedforward Bus ---
					
					% 1. Draw the connector dot
					\node[connector] at (split) {};
					
					% 2. Define the depth of the line running underneath
					\coordinate (bus_depth) at ([yshift=-0.6cm]b1.south);
					
					% 3. Draw the lines
					% Vertical down from split point
					\draw[thick] (split) -- (split |- bus_depth);
					
					% Horizontal line running across the bottom
					\draw[thick] (split |- bus_depth) -- (b3.south |- bus_depth);
					
					% Arrow UP into Block 2
					\draw[line] (b2.south |- bus_depth) -- (b2.south);
					
					% Arrow UP into Block 3
					\draw[line] (b3.south |- bus_depth) -- (b3.south);
					
				\end{tikzpicture}
			} 
			\caption{Evaluation of $\mathbf{r}(\Theta)$ given $\Theta$ and $\mathcal{D}^{\mathrm{w}}$.} 
			\label{fig:rtheta_diag}
		\end{figure}
		\section{Regularized system identification algorithm}
		\label{sec:sysID}
		Based on the aforementioned developments, we formulate the regularized system identification problem in \eqref{eq:regularized_sysID} as
		\begin{align}
			\label{eq:main_learning_problem}
			&\ \ \min_{\Theta,x_0} \ \frac{1}{N}\sum_{t=0}^{N-1} \|y_t - C x_t\|_2^2+\tau \mathbf{r}(\Theta) \\
			& \ \ \ \text{s.t.} \ 
			x_{t+1} = A(p(x_t,u_t))x_t+B(p(x_t,u_t))u_t, \ t \in \mathbb{I}_0^{N-1}, \nonumber
		\end{align}
		where $\Theta$ denotes the model parameters, $\mathbf{r}(\Theta)$ is defined in \eqref{eq:r_tracking}, and $\tau\geq 0$ is the regularization parameter. We recall that any value of $\Theta$ resulting in $\mathbb{Q}(\Theta)$ from \eqref{eq:RCI_set_parameters} being nonempty will lead to $\mathbf{r}(\Theta)<\infty$, such that the corresponding model can be used for constrained control synthesis for the plant. Since \eqref{eq:r_tracking} is itself a minimization problem, we denote its constraints as $c(\Theta,q,\mathbf{v}) \leq 0$, and append $(q,\mathbf{v})$ as optimization variables to Problem \eqref{eq:main_learning_problem} resulting in the formulation
		\begin{align}
			\label{eq:main_learning_problem_single}
			&\ \ \min_{\Theta,x_0,q,\mathbf{v}} \ \frac{1}{N}\sum_{t=0}^{N-1} \|y_t - C x_t\|_2^2+\tau \sum_{k=0}^M \|y_{k}^{\mathrm{r}}-Cz_k\|_2^2 \\
			& \ \ \ \ \  \text{s.t.} \ 
			\begin{cases}
				x_{t+1} = A(p(x_t,u_t))x_t+B(p(x_t,u_t))u_t, \\[1ex]
				z_{k+1} = A(p(z_k,v_k))z_k+B(p(z_k,v_k))v_k, \\[1ex]
				z_0=0, \ c(\Theta,q,\mathbf{v}) \leq 0, \  t \in \mathbb{I}_0^{N-1}, \  k \in \mathbb{I}_0^{M-1}.	
			\end{cases} \nonumber
		\end{align}
		While Problem \eqref{eq:main_learning_problem} can be solved using any nonlinear programming solver \cite{Nocedal2006}, the approach may become intractable if the datasets $\mathcal{D}$ or $\mathcal{D}^{\mathrm{w}}$ are large. Recall that $\mathcal{D}^{\mathrm{w}}$ is used to formulate the inequalities through \eqref{eq:W_parameters_unique} and \eqref{eq:RCI_CC} (See Figure \ref{fig:rtheta_diag}). Hence, we propose to instead penalize the inequality constraints using a soft penalty, resulting in
		\begin{align}
			\label{eq:main_learning_problem_penalized}
			&\min_{\Theta,x_0,q,\mathbf{v}} \ \frac{1}{N}\sum_{t=0}^{N-1} \|y_t - C x_t\|_2^2+\tau \sum_{k=0}^M \|y_{k}^{\mathrm{r}}-Cz_k\|_2^2\\
			& \hspace{130pt} +\tau_{\mathrm{c}} \|\max\{c(\Theta,q,\mathbf{v}),0\}^2\|_1, \nonumber
		\end{align}
		where $\tau_{\mathrm{c}}>0$ is the penalty parameter that must be chosen high enough to ensure satisfaction of the constraints. Note that the state trajectories $(x_0,\cdots,x_N)$ and $(z_0,\cdots,z_M)$ are implicit functions of $(x_0,\Theta)$ and $(\mathbf{v},\Theta)$ respectively, such that the equality constraints can be eliminated by condensing the dynamics. Since Problem \eqref{eq:main_learning_problem_penalized} is unconstrained, it can be tackled with solvers such as \textrm{Adam} \cite{kingma2017}, \textrm{L-BFGS} \cite{Byrd1995}, etc.
		While \eqref{eq:main_learning_problem_penalized} always admits a finite solution, large values indicate infeasibility of \eqref{eq:main_learning_problem_single}. {\color{black} Note that in general, Problem \eqref{eq:main_learning_problem_penalized} does not admit a unique solution, since it is a nonconvex optimization problem.} In the sequel, we present approaches to identify the scheduling variable order $n_p$, and matrix $F$ required to formulate Problem \eqref{eq:main_learning_problem_penalized}. For an approach to identify the model order $n_x$, we refer the reader to~\cite[Sec. II]{Bem24}.
	}
	\subsection{Scheduling variable order selection}
	\label{sec:np_selection}
	We present an approach to decide the scheduling variable order $n_p$ by solving an identification problem. Towards formulating this problem, we assume that the FNNs parameterizing the scheduling function in \eqref{eq:ss-p} are given by
	\begin{align}
		\label{eq:NN_split}
		\mathcal{N}_i(x,u):=W^L_i \hat{\mathcal{N}}_i(x,u)+b^L_i, && \forall i \in \mathbb{I}_1^{n_p-1}.
	\end{align}
	where $W^L_i \in \R^{1 \times n_{L}}$ and $b^L_i \in \R^{1}$ are weights and biases of the final layer of the FNNs, and $\hat{\mathcal{N}}_i:\R^{n_x+n_u}\to \R^{n_{L}}$ is the part of the FNN till the final linear layer. We note that the parameter vector $\Theta$ includes $\{(W^L_i,b^L_i),i \in \mathbb{I}_1^{n_p-1}\}$. Then, using the dataset $\mathcal{D}$, we formulate the problem
	\begin{align}
		\label{eq:np_id_problem}
		&\ \ \min_{\Theta,x_0} \ \frac{1}{N}\sum_{t=0}^{N-1} \|y_t - C x_t\|_2^2+\kappa_p \sum_{i=1}^{n_p-1} \|W^{L^{\top}}_i\|_2 \\
		& \ \ \ \text{s.t.} \ 
		x_{t+1} = A(p(x_t,u_t))x_t+B(p(x_t,u_t))u_t, \ t \in \mathbb{I}_0^{N-1}, \nonumber
	\end{align}
	where $\kappa_p>0$ is the regularization parameter.
	The regularization in \eqref{eq:np_id_problem} acts as a group-Lasso penalty, encouraging group sparsity in the vectors $W_i^L$. If $W_i^L$ is driven to zero, the corresponding FNN reduces to a constant bias $\mathcal{N}_i(x,u) = b^L_i$. In the following result, we  derive the reduced-order model under the assumption that the solution to \eqref{eq:np_id_problem} yields sparse weights. Without loss of generality, we assume the indices are ordered such that the non-zero weights appear first.
	\begin{proposition}
		\label{prop:model_reduction}
		Let the solution of~\eqref{eq:np_id_problem} be such that we have $W^L_{i} \neq 0$ for all $i \in \mathbb{I}_1^{\tilde{n}_p-1}$, and $W^L_{i}=0$ for all $i \in \mathbb{I}_{\tilde{n}_p}^{n_p-1}$. Let $\beta:=1+\sum_{i=\tilde{n}_p}^{n_p-1} e^{b^L_{i}}$ and $\tilde{\mathcal{N}}_i(x,u):=\mathcal{N}_i(x,u)-\log(\beta)$, 
		for all $i \in \mathbb{I}_1^{\tilde{n}_p-1}$. Then the system in~\eqref{eq:LPV_model} identified by~\eqref{eq:np_id_problem} can be equivalently expressed by the matrix functions
		\begin{align}
			\label{eq:reduced_matrix}
			(A(p(x,u)),B(p(x,u)))=\sum_{i=1}^{\tilde{n}_p} \tilde{p}_i(x,u)(\tilde{A}_i,\tilde{B}_i),
		\end{align}
		where $(\tilde{A}_i,\tilde{B}_i)=(A_i,B_i)$ for $i \in \mathbb{I}_1^{\tilde{n}_p-1}$ and
		\begin{align}
			&\scalemath{0.95}{
				(\tilde{A}_{\tilde{n}_p},\tilde{B}_{\tilde{n}_p}) = \frac{1}{\beta}\left(\sum_{i={\tilde{n}_p}}^{n_p-1}e^{b^L_i}(A_i,B_i) + (A_{n_p},B_{n_p})\right)},  \label{eq:AB_tilde_reg} 
		\end{align}
		and the scheduling function $\tilde{p}:\R^{n_x+n_u} \to \R^{\tilde{n}_p}$ with
		\begin{align}
			\tilde{p}_i(x,u) = \begin{cases} &\cfrac{e^{\tilde{\mathcal{N}}_i(x,u)}}{1+\sum_{j=1}^{\tilde{n}_p-1}e^{\tilde{\mathcal{N}}_j(x,u)}},\  i \in \mathbb{I}_1^{\tilde{n}_p-1},  \\
				&\cfrac{1}{1+\sum_{j=1}^{\tilde{n}_p-1}e^{\tilde{\mathcal{N}}_j(x,u)}}, \ i=\tilde{n}_p.
			\end{cases} \label{eq:p_tilde_reg}
		\end{align}
	\end{proposition}
	\begin{proof}
		By definition, for all $i \in \mathbb{I}_1^{\tilde{n}_p-1}$, we have
		\begin{align*}
			p_i(x,u) &= \frac{e^{\mathcal{N}_i(x,u)}}{\beta + \sum_{j=1}^{\tilde{n}_p-1} e^{\mathcal{N}_j(x,u)}}=\frac{e^{\tilde{\mathcal{N}}_i(x,u)}}{1 + \sum_{j=1}^{\tilde{n}_p-1} e^{\tilde{\mathcal{N}}_j(x,u)}},
		\end{align*}
		which corresponds to $\tilde{p}_i(x,u)$.
		Next, substituting $W^L_i=0$ for $i \in \mathbb{I}_{\tilde{n}_p}^{n_p-1}$ into the matrix function $A(p(x,u))$, we obtain
		\begin{align*}
			A(p(x,u)) &\scalemath{0.95}{= \sum_{i=1}^{\tilde{n}_p-1} A_i p_i(x,u) + \frac{\sum_{i=\tilde{n}_p}^{n_p-1} e^{b^L_i}A_i + A_{n_p}}{\beta + \sum_{j=1}^{\tilde{n}_p-1} e^{\mathcal{N}_j(x,u)}}} \\
			& \hspace{-0pt} \scalemath{0.95}{= \sum_{i=1}^{\tilde{n}_p-1} A_i \tilde{p}_i(x,u) +  \frac{\beta \tilde{A}_{\tilde{n}_p}}{\beta + \sum_{j=1}^{\tilde{n}_p-1} e^{\mathcal{N}_j(x,u)}}},
		\end{align*}
		where the term multiplying $\tilde{A}_{\tilde{n}_p}$ is $\tilde{p}_{\tilde{n}_p}(x,u)$. The derivation for $B(p(x,u))$ follows identically, concluding the proof.
	\end{proof}
	
	The regularization in \eqref{eq:np_id_problem} along with Proposition \ref{prop:model_reduction} present an approach to identify an equivalent qLPV model with scheduling variable order $\tilde{n}_p\leq n_p$. In some cases however, it might be desirable to reduce the order further. While this problem has been extensively studied in the literature (c.f.~\cite{olucha2024,Koelewijn2020} for a comprehensive review), the subsequent reduced-order functions might lose their simplex structure. Hence, we will now develop a procedure to reduce the order further by solving simple convex quadratic programs (QPs).
	
	\subsubsection*{A posteriori order reduction}
	\label{sec:post_processing} 
	Given a model with scheduling variable order $n_p$, the procedure in Proposition \ref{prop:model_reduction} relies on the fact that $W^L_i = 0$ forces the FNN output to a constant value $\mathcal{N}_i(x,u)=b_i^L$. This sparsity allows specific system matrices to be lumped together without loss of accuracy. We extend this intuition to an \textit{a posteriori} reduction procedure. By examining the scheduling trajectories $p(x_t,u_t)$ generated during simulation with the dataset $\mathcal{D}$, we can identify components that exhibit negligible variation. If $p_i(x_t,u_t)$ remains nearly constant for all $t \in \mathbb{I}_0^{N-1}$, the corresponding scheduling variable can be eliminated to derive an approximate reduced-order qLPV model. To formalize this intuition, we assume that the scheduling variables to be eliminated correspond to indices $i \in \mathbb{I}_{\tilde{n}_p}^{n_p-1}$ with $\tilde{n}_p<n_p$, and denote the reduced-order qLPV model that approximates \eqref{eq:LPV_model} as
	\begin{align}
		\label{eq:reduced_order_qLPV}
		\tilde{x}^+=\tilde{A}(\tilde{p}(\tilde{x},u))\tilde{x}+\tilde{B}(\tilde{p}(\tilde{x},u))u, && \tilde{y}=\tilde{C}\tilde{x},
	\end{align}
	where the matrix-valued functions are parameterized as
	\begin{align}
		\label{eq:reduced_order_linear_parameterization}
		(\tilde{A}(\tilde{p}), \tilde{B}(\tilde{p})) := \sum_{i=1}^{\tilde{n}_p}\tilde{p}_i(\tilde{A}_i, \tilde{B}_i),
	\end{align}
	and each component of the reduced-order scheduling function $\scalemath{0.95}{\tilde{p} : \R^{n_x + n_u} \to \R^{\tilde{n}_p}}$ is parameterized as
	\begin{align}
		\tilde{p}_i(\tilde{x},u) := \scalemath{0.95}{\begin{cases} &\cfrac{e^{W^L_i\hat{\mathcal{N}}_i(\tilde{x},u)}e^{\tilde{b}^L_i}}{1+\sum_{j=1}^{\tilde{n}_p-1}e^{W^L_j\hat{N}_j(\tilde{x},u)}e^{\tilde{b}^L_j}},\  i \in \mathbb{I}_1^{\tilde{n}_p-1},  \\
				&\cfrac{1}{1+\sum_{j=1}^{\tilde{n}_p-1}e^{W^L_j \hat{N}_j(\tilde{x},u)}e^{\tilde{b}^L_j}}, \ i=\tilde{n}_p.
		\end{cases}} \label{eq:reduced_order_ss-p}
	\end{align}
	The parameters of \eqref{eq:reduced_order_qLPV} are $\{\tilde{A}_i,\tilde{B}_i,i \in \mathbb{I}_1^{\tilde{n}_p}\}$, $\tilde{C}$, and  biases $\{\tilde{b}^L_i,i \in \mathbb{I}_1^{\tilde{n}_p-1}\}$ of the last layers of the FNNs, with the weights of the previous layers retained the same as \eqref{eq:LPV_model}. The reduction procedure computes these parameters to achieve
	\begin{align}
		\label{eq:reduction_goal}
		\hspace{-5pt} \scalemath{0.935}{(A(p(x_t,u_t)),B(p(x_t,u_t))) \approxeq (\tilde{A}(\tilde{p}(\tilde{x}_t,u_t)),\tilde{B}(\tilde{p}(\tilde{x_t},u_t)))},
	\end{align}
	where $x_t$ and $\tilde{x}_t$ are the state trajectories of \eqref{eq:LPV_model} and \eqref{eq:reduced_order_qLPV} when simulated from the same initial state, e.g., obtained by solving Problem \eqref{eq:np_id_problem}, using inputs from the dataset $\mathcal{D}$. To this end, we make the following assumption.
	\begin{assumption}
		\label{ass:q_eliminated}
		The components $e^{W^L_i\hat{\mathcal{N}}_i(x_t,u_t)}$ for $i \in \mathbb{I}_{\tilde{n}_p}^{n_p-1}$ exhibit negligible variation over the dataset $\mathcal{D}$.
	\end{assumption}
	
	Denoting $(A_t,B_t)=(A(p(x_t,u_t)),B(p(x_t,u_t)))$ and using the change of variables $\tilde{\mathrm{b}}_i=e^{\tilde{b}^L_i}$,
	we aim to achieve \eqref{eq:reduction_goal} under Assumption \ref{ass:q_eliminated} by solving
	\begin{align}
		\hspace{-5pt}
		\label{eq:ideal_reduction}
		\scalemath{0.94}{
			\min_{\tilde{A}_i,\tilde{B}_i,\tilde{\mathrm{b}}_i>0} \sum_{t=0}^{N-1} \left\|\begin{matrix*}[l]
				&\hspace{-10pt} (A_t,B_t)  -(\tilde{A}(\tilde{p}(\tilde{x}_t,u_t)),\tilde{B}(\tilde{p}(\tilde{x}_t,u_t)))
			\end{matrix*}\right\|_2^2,}
	\end{align}
	with the biases recovered as $\tilde{b}^L_i = \log(\tilde{\mathrm{b}}_i)$ from the optimizer. Unfortunately, Problem \eqref{eq:ideal_reduction} can be as difficult to solve as a system identification problem because of the nonlinear dynamics in \eqref{eq:reduced_order_qLPV}. To ameliorate this difficulty, we assume that
	\begin{align}
		\label{eq:second_assumption}
		W^L_i \hat{\mathcal{N}}_i(x_t,u_t) \approxeq W^L_i \hat{\mathcal{N}}_i(\tilde{x}_t,u_t),
	\end{align}
	for $i \in \mathbb{I}_1^{\tilde{n}_p-1}$ and $t \in \mathbb{I}_0^{N-1}$. Note that \eqref{eq:second_assumption} might be achieved either because $W^L_i \approxeq 0$, or the parameters $\Theta$ of \eqref{eq:LPV_model} along with structure of the FNNs. Either way, defining 
	\begin{align}
		\label{eq:qb_defn}
		(q_{it},\mathrm{\tilde{b}}_i) := \begin{cases} \left(e^{W^L_i \hat{\mathcal{N}}_i(x_t,u_t)},e^{\tilde{b}^L_i} \right), i \in \mathbb{I}_1^{\tilde{n}_p-1} \\
			(1,1), i =\tilde{n}_p, \end{cases} 
	\end{align}
	for $t \in \mathbb{I}_0^{N-1}$,
	we approximate \eqref{eq:ideal_reduction} as
	\begin{align}
		\label{eq:approx_reduction_1}
		\scalemath{0.95}{
			\min_{\tilde{A}_i,\tilde{B}_i,\tilde{\mathrm{b}}_i>0} \sum_{t=0}^{N-1} \left\|\begin{matrix*}[l]
				&\hspace{-10pt} \cfrac{\sum_{j=1}^{\tilde{n}_p}{q}_{jt} \tilde{\mathrm{b}}_j (A_t,B_t)  -\sum_{i=1}^{\tilde{n}_p}{q}_{it} \tilde{\mathrm{b}}_i (\tilde{A}_i,\tilde{B}_i)}{\sum_{j=1}^{\tilde{n}_p}{q}_{jt} \tilde{\mathrm{b}}_j}
			\end{matrix*}\right\|_2^2}
	\end{align}
	in which we exploit \eqref{eq:second_assumption} to replace the components $e^{W^L_i \hat{\mathcal{N}}_i(\tilde{x}_t,u_t)}$ in \eqref{eq:reduced_order_ss-p} with their data-based approximations $e^{W^L_i \hat{\mathcal{N}}_i(x_t,u_t)}$. While Problem \eqref{eq:approx_reduction_1} is still a nonlinear programing problem, we first note that the denominator in the objective is larger than $1$ because of \eqref{eq:qb_defn}, specifically at $j=\tilde{n}_p$. Hence, we can eliminate it from the objective to minimize an upper bound. 
	Secondly, using the change of variables
	\begin{align}
		\label{eq:change_of_variables}
		(\hat{A}_i,\hat{B}_i):=\tilde{\mathrm{b}}_i(\tilde{A}_i,\tilde{B}_i), && i \in \mathbb{I}_1^{\tilde{n}_p},
	\end{align}
	Problem \eqref{eq:approx_reduction_1} can be approximated by the QP
	\begin{align}
		\label{eq:approx_reduction_2}
		\min_{\hat{A}_i,\hat{B}_i,\tilde{\mathrm{b}}_i>0} \sum_{t=0}^{N-1} \left\|\begin{matrix*}[l]
			&\hspace{-10pt} \sum_{j=1}^{\tilde{n}_p} {q}_{jt} \tilde{\mathrm{b}}_j (A_t,B_t)  -\sum_{i=1}^{\tilde{n}_p} {q}_{it}  (\hat{A}_i,\hat{B}_i)
		\end{matrix*}\right\|_2^2.
	\end{align}
	Hence, we propose to solve Problem \eqref{eq:approx_reduction_2}, and recover the system matrices $(\tilde{A}_i,\tilde{B}_i)$ as in \eqref{eq:change_of_variables} and biases as $\tilde{b}^L_i = \log(\tilde{\mathrm{b}}_i)$. Then, using the same initial state, we simulate the dynamics in \eqref{eq:reduced_order_qLPV}, and estimate matrix $\tilde{C}$ by solving the QP
	\begin{align}
		\label{eq:approx_reduction_3}
		\min_{\tilde{C}} \sum_{t=0}^{N-1} \left\|\begin{matrix*}[l]
			&\hspace{-10pt} y_t -\tilde{C}\tilde{x}_t
		\end{matrix*}\right\|_2^2,
	\end{align}
	where $\tilde{x}_t$ is the state of \eqref{eq:reduced_order_qLPV} with the newly estimated matrices using inputs in the dataset $\mathcal{D}$, and $y_t$ the corresponding plant outputs in the same dataset. Thus, by solving the QPs in \eqref{eq:approx_reduction_2} and \eqref{eq:approx_reduction_3}, we identify a qLPV system with scheduling variable order $\tilde{n}_p<n_p$ to approximate \eqref{eq:LPV_model}.
	\begin{remark}
		While Problem \eqref{eq:approx_reduction_2} is derived for a single dataset $\mathcal{D}$, it can straightforwardly be extended to accommodate multiple datasets.
	\end{remark}
	\begin{remark}
		Scheduling variable order reduction can also be achieved by observing the variance of matrices $(A_i,B_i)$. If two such pairs are the same, then they can be lumped together. The development of such techniques is a subject of future study.
	\end{remark}
	\subsection{Parameterization of the RCI set}
	We present an approach to compute a matrix $F$ given initial model parameters $\Theta$ such that the set $\mathbb{Q}(\Theta)$ in \eqref{eq:RCI_set_parameters} is nonempty. We note that these parameters $\Theta$ are not expected to capture the plant dynamics well. Rather, they are used to compute a matrix $F$ that characterizes the RCI set $\mathbb{X}(q)$, and will be refined later by solving Problem \eqref{eq:main_learning_problem_penalized}.
	This is a standard approach in many data-driven robust control synthesis procedures, e.g., \cite{Chen2022,Mejari2024}, etc. For given $\Theta$, there exist many techniques to compute an initial matrix $F$, e.g., \cite{Pluymers2005}. The approach we adopt, however, explicitly permits the user to trade-off complexity for conservativeness. 
	
	Given $\Theta$, we first compute the disturbance set parameters $(c_{\mathrm{w}},\epsilon_{\mathrm{w}})$ as in \eqref{eq:W_parameters_unique}. We remark that since $\Theta$ is expected to induce a large prediction error, the set $\mathbb{W}_{\kappa}(c_{\mathrm{w}},\epsilon_{\mathrm{w}})$ might be large.
	Then, we assume to be given a matrix $\tilde{F} \in \R^{\fe \times n_x}$ that characterizes a polytope with $\ve$ vertices as
	\begin{align}
		\label{eq:tilde_F}
		\{x \in \R^{n_x} : \tilde{F}x \leq 1\}=\mathrm{CH}\{\tilde{x}_l, l \in \mathbb{I}_1^{\ve}\},
	\end{align}
	Typically, $\tilde{F}$ is chosen such that $\{x \in \R^{n_x}:\tilde{F}x\leq 1\}$ is the $\infty$-norm box in $\R^{n_x}$ with $\fe = 2n_x$ and $\ve=2^{n_x}$. The following result from \cite{mulagaleti2023} can be used to compute a matrix $F$ characterizing the RCI set $\mathbb{X}(q)$ for the uncertain LTI system in \eqref{eq:multiplicative_RCI}, and consequently the uncertain qLPV system in \eqref{eq:uncertain_qLPV} as per Proposition \ref{prop:RCI_transfer}.
	\begin{proposition}
		\label{prop:init_RCI}
		Suppose there exists an invertible matrix $\Sigma \in \R^{n_x \times n_x}$ and vertex control inputs $\mathrm{u}:=(\mathrm{u}_1,\cdots,\mathrm{u}_{\ve})$ satisfying the nonlinear inequalities
		\begin{subequations}
			\label{eq:RCI_init}
			\begin{align}
				\hspace{-3pt}
				\scalemath{0.99}{\tilde{F}\Sigma^{-1}(A_i \Sigma \tilde{x}_l + B_i \mathrm{u}_l+K_i c_{\mathrm{w}}) + \kappa|F\Sigma^{-1}K_i|\epsilon_{\mathrm{w}} \leq 1,} \label{eq:RCI_init_1}\\
				\scalemath{0.99}{ H^y (C \Sigma \tilde{x}_l + c_{\mathrm{w}})+ \kappa|H^y| \epsilon_{\mathrm{w}} \leq h^y,} \label{eq:RCI_init_2} \\
				\scalemath{0.99}{H^u \mathrm{u}_l \leq h^u,} \label{eq:RCI_init_3}
			\end{align}
		\end{subequations}
		for all $i \in \mathbb{I}_1^{n_p}$ and $l \in \mathbb{I}_1^{\ve}$. Then, the set $\mathbb{X}(1)$ defined with $F=\tilde{F}\Sigma^{-1}$ is RCI for the uncertain LTI system \eqref{eq:multiplicative_RCI}.
	\end{proposition}
	
	Following Proposition \ref{prop:init_RCI} and recalling that our goal is to synthesize controllers that track output reference signals of the form $\{y^{\mathrm{r}}_k \in \R^{n_y},k \in \mathbb{I}_0^{M}\}$, we propose to compute $\Sigma$ by solving the nonlinear programming problem
	\begin{align}
		\label{eq:NLP_init}
		& \ \min_{\Sigma, \mathrm{u}, \mathbf{v}} \ \sum_{k=0}^M \|y_k^{\mathrm{r}}-Cz_k\|_2^2 \\
		& \ \ \text{s.t.} \ 
		\begin{cases}
			\eqref{eq:RCI_init},\  z_0 = 0, \ v_k \in \mathbb{U}, \ \tilde{F}\Sigma^{-1}z_k \leq 1,\\[1ex]
			z_{k+1} = A(p(z_k,v_k))z_k + B(p(z_k,v_k))v_k, \ k \in \mathbb{I}_0^{M-1},
		\end{cases} \nonumber
	\end{align}
	formulated in a similar way as \eqref{eq:r_tracking}. We remark that Problem \eqref{eq:NLP_init} simplifies significantly if the initial model is LTI, i.e., $(A_i,B_i,K_i)=(\bar{A},\bar{B},\bar{K})$ for all $i$. In such cases, the constraints in \eqref{eq:RCI_init} need only be enforced for a single system, and the dynamics constraint in \eqref{eq:NLP_init} reduces to $\bar{x}^+=\bar{A}\bar{x}+\bar{B}\bar{u}$. Assuming that \eqref{eq:NLP_init} is feasible, the matrix $F$ is computed as $F=\tilde{F}\Sigma^{-1}$, following which the configuration-constraint matrix $E$ and vertex maps $\{V_l,l \in \mathbb{I}_1^{\ve}\}$ are computed following \cite[Theorem 2]{Villaneuva2024}.
	
	Given matrix $\tilde{F}$ satisfying \eqref{eq:tilde_F}, we now present a heuristic approach to compute system parameters $\Theta$ that encourage feasibility of Problem \eqref{eq:NLP_init}. We attempt to compute model parameters $\Theta$ such that $\Sigma=I_{n_x}$, i.e, the identity matrix, satisfies \eqref{eq:RCI_init}. To this end, we formulate the problem
	\begin{align}
		\label{eq:init_learning_problem}
		&\ \ \min_{\Theta,x_0,\mathrm{u}} \ \frac{1}{N}\sum_{t=0}^{N-1} \|y_t - C x_t\|_2^2\\
		& \ \ \ \ \  \text{s.t.} \ 
		\begin{cases}
			x_{t+1} = A(p(x_t,u_t))x_t+B(p(x_t,u_t))u_t, \\[1ex]
			\tilde{F}(A_i \tilde{x}_l + B_i \mathrm{u}_l+K_i c_{\mathrm{w}}) + \kappa|FK_i|\epsilon_{\mathrm{w}} \leq 1, \\[1ex]
			H^y (C \Sigma \tilde{x}_l + c_{\mathrm{w}})+ \kappa|H^y| \epsilon_{\mathrm{w}} \leq h^y, \\[1ex]
			H^u \mathrm{u}_l \leq h^u, i \in \mathbb{I}_1^{n_p}, l \in \mathbb{I}_1^{\ve}, t \in \mathbb{I}_1^{N-1}.
		\end{cases} \nonumber
	\end{align}
	Note that Problem \eqref{eq:init_learning_problem} is formulated be posing the conditions in \eqref{eq:RCI_init} formulated with $\Sigma=I_{n_x}$ as constraints. To solve it, we relax the inequality constraints with a large penalty parameter. In practice, this approach was observed to compute models $\Theta$ that render Problem \eqref{eq:NLP_init} feasible.
	
	\begin{algorithm}[t]
		\caption{Procedure to solve Problem~\eqref{eq:main_learning_problem}}\label{alg:procedure}
		\begin{algorithmic}[1]
			\Require Dataset $\mathcal{D}$ and $\mathcal{D}^{\mathrm{w}}$, robustness parameter $\kappa$
			\State Solve Problem \eqref{eq:np_id_problem} to estimate $n_p$, label optimal model parameters as $\Theta^1$.
			\State (Optional) Solve Problem \eqref{eq:approx_reduction_2} to reduce $n_p$ further. 
			\State Solve Problem \eqref{eq:init_learning_problem} with relaxed constraints, initialized at $\Theta^1$. Label optimal model parameters a as $\Theta^2$.
			\State Solve Problem \eqref{eq:NLP_init} with $\Theta=\Theta^2$ to compute $F$. 
			\State Solve Problem \eqref{eq:main_learning_problem_penalized} initialized at $\Theta^2$. Label optimal model parameters as $\Theta^3$.
			\State If $\mathbb{Q}(\Theta^3)$ is nonempty, return $\Theta^3$. Else, re-solve Problem \eqref{eq:main_learning_problem_penalized} with larger $\tau_{\mathrm{c}}$.
		\end{algorithmic}
	\end{algorithm}
	
	\subsection{qLPV SysID with RCI regularization}
	\label{sec:concurrent_qLPV_identification}
	In Algorithm \ref{alg:procedure}, we outline the procedure to formulate and solve Problem \eqref{eq:main_learning_problem}. Firstly, we follow the procedure in Section \ref{sec:np_selection} to identify the scheduling variable order $n_p$. Then, we solve Problem \eqref{eq:init_learning_problem} with relaxed constraints, using a solver initialized with the solution of the previous step. We use the computed model parameters $\Theta^2$ to solve the nonlinear programming problem \eqref{eq:NLP_init} to compute $F$, using which we compute matrices $E$ and $V$ required to formulate the conditions in \eqref{eq:RCI_CC} and characterize the set $\mathcal{Q}(\Theta)$. Note that if Problem \eqref{eq:NLP_init} is feasible with $\Theta=\Theta^2$, then $\mathcal{Q}(\Theta^2)$ is nonempty. Finally, we solve Problem \eqref{eq:main_learning_problem_penalized}, with the solver initialized at $\Theta^2$. We choose the penalty parameter $\tau_c$ to be large enough that the constraints of Problem \eqref{eq:main_learning_problem} are satisfied at the solution, or equivalently, the set $\mathbb{Q}(\Theta^3)$ is nonempty thus solving Problem \ref{prob:basic_problem}.
	
	{\color{black}
		\subsection{Challenges and limitations}
		We identify the following challenges in the implementation of Algorithm \ref{alg:procedure}.
		\subsubsection{Trajectory Length} Using datasets $\mathcal{D}$ with long trajectories (large $N$) can make the identification problem difficult to solve due to the instability of the dynamics during optimization \cite{Ribeiro2020}. For our qLPV parameterization in \eqref{eq:LPV_model}, a reasonable heuristic to mitigate this involves first identifying a stable LTI system $x^+=Ax+Bu$ that best fits the data, and initializing  $(A_i,B_i)=(A,B)$ for all $i \in \mathbb{I}_1^{n_p}$ such that the qLPV system behaves as a stable LTI system for all $p \in \mathcal{P}$.
		\subsubsection{Robustness Parameter} Selecting the disturbance set parameter $\kappa$ involves a trade-off. As discussed previously, while larger values yield more robust models, they may render the RCI set synthesis infeasible. 
		\subsubsection{Initial RCI Set} While the model identified by solving Problem \eqref{eq:init_learning_problem} with relaxed constraints was generally observed to result in feasibility of Problem \eqref{eq:NLP_init} to compute the matrix $F$, this is not guaranteed. Furthermore, Problem \eqref{eq:init_learning_problem} might fail simply due to unstable and uncontrollable modes in the data, rendering the original Problem \eqref{prob:basic_problem} inherently infeasible. Distinguishing between numerical infeasibility and this fundamental structural infeasibility remains a subject for future study.	From a numerical standpoint, feasibility can be improved by choosing a more expressive template matrix $\tilde{F}$. However, this may result in a polytope $\mathbb{X}(q)$ with a large number of vertices, significantly increasing the complexity of the set $\mathbb{Q}(\Theta)$ and Problem \eqref{eq:main_learning_problem}. An alternative is to co-synthesize a linear feedback controller for the uncertain LTI system in \eqref{eq:multiplicative_RCI} to obtain a low-complexity RCI set, though this typically reduces control authority. 
		\subsubsection{Scalability} Since our approach relies on polytopic RCI sets, scalability with the state dimension $n_x$ remains a concern, since the complexity of a polytope in terms of number of facets and vertices can increase exponentially with the dimension. An approach to tackle this complexity is to carefully compute the template matrix $F$ using a structured approach, such as that presented in \cite{Badalamenti2025}.
	}

	\section{Numerical results}
	\label{sec:numerical_examples}
	\subsection{Effectiveness of the qLPV parameterization}
	We want to illustrate the effectiveness of the proposed parameterizations for system identification and control synthesis through numerical examples. First, we want to show the effectiveness of the qLPV parameterization, in which we solve Problem~\eqref{eq:main_learning_problem} with the constraints ignored, i.e.,
	\begin{align}
		\label{eq:SysID_qLPV_examples}
		&\min_{\Theta,x_0} \ \frac{1}{N}\sum_{t=0}^{N-1} \|y_t - C x_t\|_2^2 \\
		& \ \ \text{s.t.} \  x_{t+1}=A(p(x_t,u_t))x_t+B(p(x_t,u_t))u_t, \ t \in \mathbb{I}_0^{N-1}. \nonumber
	\end{align}
	We solve Problem~\eqref{eq:SysID_qLPV_examples} using the \url{jax-sysid} library~\cite{Bem24} with scaled data $u \leftarrow (u-\mu_u)/\sigma_u$ and $y \leftarrow (y-\mu_y)/\sigma_y$, where $\mu_u,\mu_y$ and $\sigma_u,\sigma_y$ are the empirical mean and standard deviation of the inputs and outputs in the Training dataset.
	The quality of fit is measured using the best fit ratio (BFR)
	\begin{align*}
		\scalemath{0.95}{
			\text{BFR} := \max\left\{0,1-\sqrt{\frac{\|y_t-\hat{y}_t\|^2_2}{\|y_t-(1/N)\sum_{t=0}^{N-1} y_t\|^2_2}}\right\}\times 100,}
	\end{align*}
	where $y_t$ and $\hat{y}_t$ are the measured and predicted outputs respectively. When evaluating BFR, we use the entire dataset to estimate the initial state of the qLPV model. For all examples, we solve Problem~\eqref{eq:SysID_qLPV_examples} by first running $2000$ Adam iterations~\cite{kingma2017}, followed by a maximum of $5000$ iterations of the L-BFGS-B solver~\cite{Byrd1995}. We start the optimization with a randomly initialized $\Theta$ \footnote{Corresponding code found on \url{https://github.com/samku/qLPV_concurrent_identification}}.

	\subsubsection{Trigonometric system} 
	\label{ex1:qLPV}
	We consider the problem of identifying a qLPV model of the nonlinear plant
	\begin{align*}
		%\label{eq:trigonometric_system}
		z^+ &= \begin{bmatrix}
			a_1\sin(z_1)+b_1\cos(0.5z_2)u \\
			a_2\sin(z_1+z_3)+b_2\arctan(z_1+z_2) \\
			a_3e^{-z_2} + b_3 \sin(-0.5z_1)u
		\end{bmatrix} + q_x, \nonumber  \\
		y &= \arctan(c_1z_1^3) + \arctan(c_2z_2^3) + \arctan(c_3 z_3^3) + q_y 
	\end{align*}
	where the additive noise terms $q_x \sim 0.01\mathcal{N}^3$, $q_y \sim 0.01 \mathcal{N}$, and $a = [0.5,0.6,0.4]$, $b = [1.7,0.4,0.9]$, $c = [2.2,1.8,-1]$. The training and test datasets, consisting of $5000$ points each, are built by exciting the system using inputs $u$ uniformly distributed in $[-0.5,0.5]$. To identify a qLPV model, we solve Problem~\eqref{eq:SysID_qLPV_examples} with $n_x = 3$, scheduling variable dimension $n_p \in \{1,2,3,4\}$, and number of hidden layers $n_h \in \{1,2\}$ for each FNN $\mathcal{N}_i(x,u)$ defining the scheduling function as in~\eqref{eq:ss-p}. 
	We use $6$ $\mathrm{swish}$ activation units in each layer. 
	The obtained BFR scores for $n_p=1$, which corresponds to identifying an LTI system with $n_x=3$ are $67.962$ over the training set and $68.449$ over the test set, indicating a poor fit quality. 
	For qLPV identification, the BFR scores are shown in Table~\ref{trigonometric}. 
	We observe that using the proposed qLPV parameterization improves the fit quality. 
	
	In order to study the effect of group Lasso regularization in \eqref{eq:np_id_problem} for reducing the scheduling variable order, we solve Problem \eqref{eq:np_id_problem} with $n_x=2$, $n_p=10$, $n_h=1$, and the FNN $\mathcal{N}_i(x,u)$ parameterized with $3$ $\mathrm{swish}$ activation units. In Figure \ref{fig:glasso_effect}, we plot the BFR scores and reduced scheduling variable order as $\kappa_p$ increases. We report that $\tilde{n}_p=10$ for all $\kappa_p \in [0,10^{-6}]$. For $\kappa_p \geq 2.75 \times 10^{-6}$, we obtain $\tilde{n}_p=1$, such that the qLPV system reduces to an LTI system.
	\iffalse
	\begin{table}
		\centering
		\begin{tabular}{|l|l|l|l|l|}
			\hline
			& \multicolumn{2}{c|}{$n_h=1$} & \multicolumn{2}{c|}{$n_h=2$}  \\
			\cline{2-5}
			$n_p$ & Training & Test & Training & Test \\
			\hline
			2                               & 93.426 &     94.060                    & 94.259 &      93.262                     \\
			\hline
			3                               & 95.154 &     96.022                 & 95.995 &       95.223                 \\
			\hline
			4                               & 95.978 &         96.331               & 96.120 &         96.143                 \\
			\hline 
		\end{tabular}
		\caption{BFR scores for the trigonometric system \\ with state order $n_x=3$. \label{trigonometric}}
	\end{table}
	\fi
	\begin{table}
		\centering
		% Set width to \columnwidth (or \textwidth)
		\begin{tabularx}{\columnwidth}{|Y|Y|Y|Y|Y|}
			\hline
			& \multicolumn{2}{c|}{$n_h=1$} & \multicolumn{2}{c|}{$n_h=2$} \\
			\cline{2-5}
			$n_p$ & Training & Test & Training & Test \\
			\hline
			2 & 93.426 & 94.060 & 94.259 & 93.262 \\
			\hline
			3 & 95.154 & 96.022 & 95.995 & 95.223 \\
			\hline
			4 & 95.978 & 96.331 & 96.120 & 96.143 \\
			\hline 
		\end{tabularx}
		\caption{BFR scores for the trigonometric system with $n_x=3$. \label{trigonometric}}
	\end{table}
	\begin{figure}
		\centering
		\includegraphics[width=1.\linewidth, clip=true]{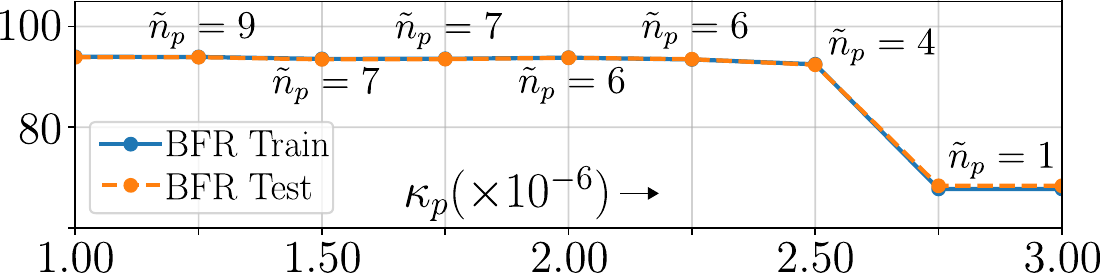}
		\captionsetup{width=\linewidth,justification=justified, singlelinecheck=false}
		\caption{Variation of BFR with group Lasso regularization parameter $\kappa_{\mathrm{p}}$ in Problem~\eqref{eq:np_id_problem} with $n_p=10$, along with reduced scheduling variable $\tilde{n}_p$ order following Proposition \ref{prop:model_reduction}. We obtain $\tilde{n}_p=10$ for $\kappa_p \in [0,10^{-6}]$ and $\tilde{n}_p=1$ for $\kappa_p \geq 2.75 \times 10^{-6}$.}
		\label{fig:glasso_effect}
	\end{figure}
	\subsubsection{Benchmark SysID examples}
	We solve Problem~\eqref{eq:SysID_qLPV_examples} using datasets generated by benchmark systems. The BFR scores are shown in Table~\ref{benchmark}. In all examples, we set $n_x=4$ and $n_p=3$, and parameterize each FNN $\mathcal{N}_i(x,u)$ with a single hidden later containing $6$ $\mathrm{swish}$ activation units. Further, we append the regularization $1e^{-4} \|\Theta\|_2^2$ to improve generalization quality of the identified model. The datasets are indicated as
	\begin{itemize}
		\item
		$[\mathrm{2T}]$  Two-tank system from~\cite{MathWorks2023TwoTank}, with $2300$ training and $700$ test data. 
		\item
		$[\mathrm{SB}]$  Silverbox system from~\cite{Ljung2020}, with $7950$ training and $1023$ test data. 
		\item $[\mathrm{HW}]$ Hammerstein-Wiener system from~\cite{schoukens2009wiener}, with $75000$ training and $108000$ test data. 
		\item
		$[\mathrm{MR}]$ Magneto-Rheological Fluid Damper system from~\cite{Wang2009}, with $2000$ training and $1499$ 
		%\AB{1499?} 
		test data. 
	\end{itemize}
	\begin{table}
		\centering
		% Set width to \columnwidth to fill the page width
		\begin{tabularx}{\columnwidth}{|Y|Y|Y|Y|Y|}
			\hline
			& \multicolumn{2}{c|}{$n_p=1$} & \multicolumn{2}{c|}{$n_p=3$} \\
			\cline{2-5}
			& Training & Test & Training & Test \\
			\hline
			$\mathrm{2T}$ & 78.820 & 83.336 & 97.152 & 96.709 \\
			\hline
			$\mathrm{SB}$ & 78.904 & 50.201 & 98.690 & 98.419 \\
			\hline
			$\mathrm{HW}$ & 82.449 & 82.111 & 95.939 & 95.089 \\
			\hline 
			$\mathrm{MR}$ & 56.231 & 49.771 & 92.717 & 91.464 \\
			\hline 
		\end{tabularx}
		\caption{BFR scores for the considered benchmark examples. \label{benchmark}}
	\end{table}
	\iffalse
	\begin{table}
		\centering
		\begin{tabular}{|l|l|l|l|l|l|}
			\hline
			& \multicolumn{2}{c|}{$n_p=1$} & \multicolumn{2}{c|}{$n_p=3$}  \\
			\cline{2-5}
			& Training & Test & Training & Test \\
			\hline
			$\mathrm{2T}$                                            &78.820 &     83.336                   & 97.152 &     96.709                \\
			\hline
			$\mathrm{SB}$                                            &  78.904 &      50.201                   & 98.690 &      98.419                 \\
			\hline
			$\mathrm{HW}$                                            &82.449 &        82.111             &95.939 &       95.089               \\
			\hline 
			$\mathrm{MR}$                                           & 56.231 &        49.771               & 92.717 &         91.464               \\
			\hline 
		\end{tabular}
		\caption{BFR scores for the considered \\ benchmark examples. \label{benchmark}}
	\end{table}
	\fi
	In Table \ref{benchmark}, we also indicate the BFR scores when identifying an LTI model, i.e., with $n_p=1$. Observe that the qLPV parameterization achieves a good predictive performance.
	\begin{figure}
		\centering
		\includegraphics[width=1.\linewidth, clip=true]{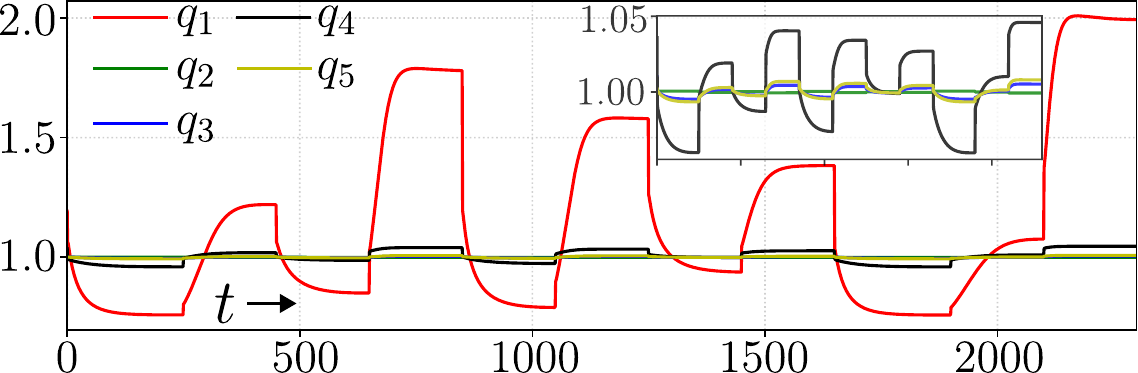}
		\captionsetup{width=\linewidth,justification=justified, singlelinecheck=false}
		\caption{Trajectories of $q_{it}$ defined in \eqref{eq:qb_defn} for the Two-tank system simulated with the training dataset. Observe negligible variation in all except the trajectories of $q_1$, thus verifying Assumption \ref{ass:q_eliminated}. The smaller plot is zoomed in on the trajectories of the remaining components.}
		\label{fig:postprocessing}
	\end{figure}
	
	For the two-tank system ($\mathrm{2T}$), we now illustrate the a posteriori scheduling order reduction procedure proposed in Section \ref{sec:post_processing}. We consider the over-parameterized qLPV model with $n_x=4$ and $n_p=6$, with each FNN parameterized with a single hidden layer containing $6$ $\mathrm{swish}$ activation units. 
	With this model, we obtain BFR scores of $94.194$ and $92.625$ over the training and test sets respectively. Towards reducing the scheduling variable order of this model, we plot the values of $q_{it}$ defined in \eqref{eq:qb_defn} when simulated using the dataset $\mathcal{D}$ in Figure \ref{fig:postprocessing}. We observe that $q_{1}$ exhibits the largest variation over the dataset, with the others being almost constant thus verifying Assumption \ref{ass:q_eliminated}. Hence, we eliminate those that are almost constant and estimate the parameters of the reduced-order model in \eqref{eq:reduced_order_qLPV} by solving the QPs \eqref{eq:approx_reduction_2} and \eqref{eq:approx_reduction_3}. In Table \ref{tab:bfr_reduction}, we see the BFR scores of the reduced order models, with $\iota$ denoting the eliminated indices. As expected, upto $3$ orders can be eliminated without any degradation in the model quality.
	
	\begin{table}[htbp]
		\centering
		\renewcommand{\arraystretch}{1.2}
		\setlength{\tabcolsep}{3pt} % Reduces padding between columns to save space
		\begin{tabularx}{\columnwidth}{|X|c|c||X|c|c|}
			\hline
			$\iota$ & Training & Test & $\iota$ & Training & Test \\
			\hline
			\hline
			$\{2\}$       & 94.196 & 92.599 & $\{2,3\}$     & 94.197 & 92.586 \\
			\hline
			$\{2,3,5\}$   & 94.160 & 92.741 & $\{2,3,5,4\}$ & 88.930 & 90.005 \\
			\hline
		\end{tabularx}
		\caption{BFR values for reduced order qLPV model.}
		\label{tab:bfr_reduction}
	\end{table}

	\subsection{SysID with controller synthesis guarantees}
	We consider data generated by the nonlinear spring-mass-damper system with dynamics
	\begin{align*}
		\mathrm{m}_{[1]}\ddot{\mathrm{x}}_{[1]}&=10 u-\mathrm{k}^{\mathrm{s}}(\mathrm{x}_{[1]})-\mathrm{k}^{\mathrm{d}}(\dot{\mathrm{x}}_{[1]})-\mathrm{k}^{\mathrm{s}}(\delta \mathrm{x})-\mathrm{k}^{\mathrm{d}}(\dot{ \delta \mathrm{x}}), \\
		\mathrm{m}_{[2]}\ddot{\mathrm{x}}_{[2]}&=-\mathrm{k}^{\mathrm{s}}(\mathrm{x}_{[2]})-\mathrm{k}^{\mathrm{d}}(\mathrm{x}_{[2]})+\mathrm{k}^{\mathrm{s}}(\delta \mathrm{x})+\mathrm{k}^{\mathrm{d}}(\dot{ \delta \mathrm{x}}),
	\end{align*}
	where $(\mathrm{x}_{[1]},\mathrm{x}_{[2]})$ are the positions of the masses. The input $u$ is the force applied on the first mass, and the output is the position of the second mass. We denote the relative position of the masses as $\delta \mathrm{x}=\mathrm{x}_{[1]}-\mathrm{x}_{[2]}$, and the spring and damper reaction forces are $\mathrm{k}^{\mathrm{s}}(x)=\mathrm{a}x+\mathrm{b}x^3$ and $\mathrm{k}^{\mathrm{d}}(v)=\mathrm{d}v+\mathrm{e} \cdot \mathrm{tanh}(v/\mathrm{v}_\mathrm{0})$ respectively. The plant parameters are $(m_{[1]},m_{[2]},\mathrm{a},\mathrm{b},\mathrm{d},\mathrm{e},\mathrm{v}_\mathrm{0})=(0.25,0.1,1,1,0.5,0.5,0.01)$ in appropriate units. The system in simulated from the origin using the Runge-Kutta integrator (Tsit5) implemented in the $\mathrm{Diffrax}$ library \cite{diffrax}, with a time-step of $0.02$s using inputs $u \in \mathbb{U}=[-1,1]$. We excite the system using a combination of multisine and random inputs in this range, using which we build the training dataset $\mathcal{D}$, observer dataset $\mathcal{D}^{\mathrm{w}}$ and test dataset, each with $5000$ points. For all these sets, we initialize the plant at the origin. Using the outputs reached by the plant within these datasets, we define the output constraint set as $\mathbb{Y}=[-1.645,2.071]$. We note that the system is nonlinear, indicated by the fact that the best identified LTI system with $n_x=4$ achieves BFR scores of $77.131$, $70.812$ and $70.206$ on the training, observer and test datasets respectively.
	\begin{figure}
		\centering
		\includegraphics[width=1.\linewidth, clip=true]{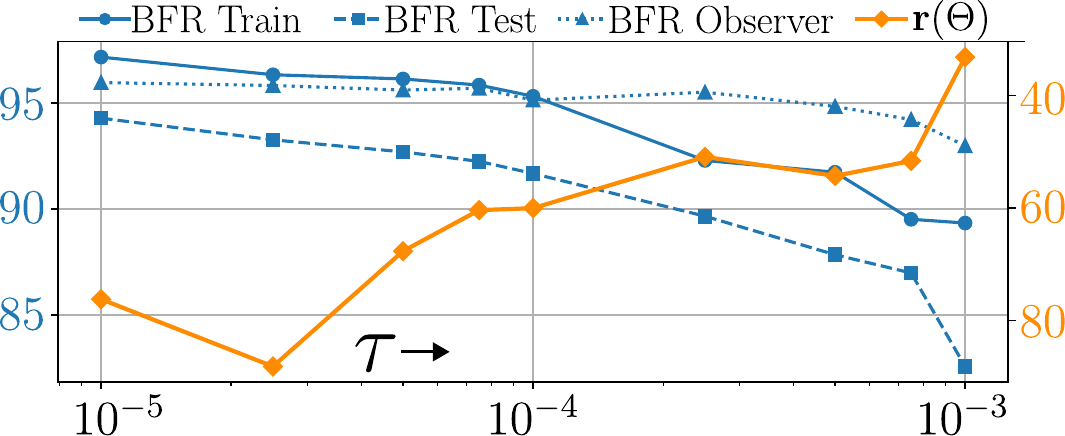}
		\captionsetup{width=\linewidth,justification=justified, singlelinecheck=false}
		\caption{Variation of BFR and RCI set size $\mathbf{r}(\mathbf{x})$ with $\tau$. Observe a tradeoff as $\tau$ varies. We select $\tau=10^{-4}$ for controller synthesis.}
		%\AB{Can we invert the red plot, showing 50 on top and 350 on bottom?}}
	\label{fig:tau_effect}
\end{figure}
For this system, we aim to design an output reference tracking controller developed using a qLPV model with reduced state dimension $n_x=3$. To this end, we perform the steps outlined in Algorithm \ref{alg:procedure}. We first identify that $n_p=3$ is a suitable model order. Then, we parameterize each FNN $\mathcal{N}_i(x,u)$ in \eqref{eq:ss-p} with a single hidden layer and $4$ $\mathrm{swish}$ activation units. To identify an initial model $\Theta^1$, we solve Problem \eqref{eq:init_learning_problem} with relaxed constraints and $\kappa=1.1$, in which we select $\tilde{F}=[I_3 \ \ -I_3]^{\top}$. In this problem, we set the observer gains $K_i=0$, such that the disturbance set is identified purely on the output. While larger values of $\kappa$ might be selected to increase robustness, $\kappa=1.1$ was found sufficient for our simulations. Then, we solve Problem \eqref{eq:NLP_init} using the IPOPT solver \cite{Wachter2006} interfaced through CasADi \cite{CasADi}, in which we set $M=50$, and formulate it with two reference trajectories $y^{\mathrm{r}}_k=2.071$ and $y^{\mathrm{r}}_k=-1.645$ for all $k \in \mathbb{I}_0^M$ $-$ We compute an RCI set in which our system can be driven from the origin to the vertices of $\mathbb{Y}$. 
The resulting matrix $F$ that parameterizes the RCI set $\mathbb{X}(q)$ is given as $F=\tilde{F} \Sigma^{-1}$, where
\begin{align*}
	\Sigma^{-1} = \begin{bmatrix*}[c]
		0.3098 & 2.7690 & 3.0873 \\
		0.1937 & 0.5118 & -0.1277 \\
		-0.9419 & 0.3280 & 0.6171
	\end{bmatrix*}
\end{align*}
This results in $\mathbb{X}(q)$ having $\mathbf{f}=6$ facets and $\ve=8$ vertices. For $\mathbb{X}(1)$, we compute the matrices $E$ and $V$ following the procedure in \cite{Villaneuva2024}, using which we formulate Problem \eqref{eq:main_learning_problem_penalized} with penalty parameter $\tau_c = 1000$. This penalty was found to be large enough for constraint satisfaction. In Figure \ref{fig:tau_effect}, we plot the variation in BFR scores, along with the value of the regularization function $\mathbf{r}(\Theta)$ at the solution of Problem \eqref{eq:main_learning_problem_penalized} for different values of $\tau$, with smaller values of $\mathbf{r}(\Theta)$ indicating a less conservative RCI set, as stated in \eqref{eq:r_reqmts}. In this plot, BFR Observer refers to the BFR score for the closed-loop observer model \eqref{eq:closed_loop_observer} over the dataset $\mathcal{D}^{\mathrm{w}}$. 
We note that as $\tau$ increases, conservativeness of the RCI set reduces at the expense of model quality. For $\tau\geq 10^{-4}$ , we notice that the BFR observer score sustains at a high value. This is because the model tends to compute large values of the observer gains $\{K_1,K_2,K_3\}$ to minimize the value of $\epsilon_{\mathrm{w}}$ defined in \eqref{eq:W_parameters_unique} in order to reduce the disturbance set size. This, however comes at the price of open-loop prediction inaccuracy as seen by the degradation of BFR Train and BFR Test scores . {\color{black}For comparison, we report that the \textit{sequential design} procedure, which involves first identifying a qLPV model by solving Problem \eqref{eq:SysID_qLPV_examples}, and then computing an RCI set, results in BFR scores of $98.175$, $88.686$ and $94.002$ over the training, test and Observer datasets, with a high conservativeness value of $\mathbf{r}(\Theta)=349.977$. This result validates the usage of concurrent synthesis to solve Problem \ref{prob:basic_problem} with reduced conservativeness.}

Using the model identified with $\tau=10^{-4}$, we then design an output feedback controller for the underlying plant. Given the current plant output $y_t$, observer state $z_t$, and output reference $y_t^{\mathrm{r}}$, we formulate our controller as
\begin{align}
	\label{eq:feedback_controller}
	(u_t^*,q_t^*)&:= \arg\min_{u,q} \|y_t^{\mathrm{r}}-Cz^+\|_2^2 \\
	& \ \ \ \ \ \ \text{s.t.} \begin{cases} \nonumber
		z^+=A(p)z_t+B(p)u+K(p)(y_t-Cz_t), \\[1ex]
		p=p(z_t,u), \ u \in \mathbb{U}, \\[1ex]
		q \in \mathbb{Q}(\Theta), \ z^+ \in \mathbb{X}(q),
	\end{cases}
\end{align}
where $\mathbb{Q}(\Theta)$ is defined in \eqref{eq:RCI_set_parameters}. In Problem \eqref{eq:feedback_controller}, we seek to find a feasible control input $u$ and an RCI set $\mathbb{X}(q)$ in which the subsequent state belongs, while minimizing the distance between subsequent output and the reference. 
Recall that our model $\Theta$ is constructed such that $\mathbb{Q}(\Theta)$ is nonempty. Furthermore, under Assumption \ref{ass:kappa_assumption}, we always have $y_t-Cz_t \in \mathbb{W}_{\kappa}(\mathrm{c}_{\mathrm{w}},\epsilon_{\mathrm{w}})$. Then, according to Proposition \ref{prop:RCI_exists}, if Problem \eqref{eq:feedback_controller} is initialized with $\hat{x}_0-z_0 \in \mathcal{E}$ (where $\hat{x}_0$ is the state of the underlying plant \eqref{eq:underlying_fake}), then the closed-loop system of the observer \eqref{eq:closed_loop_observer} and the plant is recursively feasible with $u_t=u_t^*$. Since the scheduling variable depends on $u$, Problem \eqref{eq:feedback_controller} is a nonlinear optimization problem. By parameterizing the FNNs $\mathcal{N}_i(z,u)$ to be independent of $u$, we can obtain a QP formulation instead. Note that since $\mathbb{Q}(q)$ is nonempty, one can instead design tube-based MPC controllers, e.g., \cite{Badalamenti2024} that are also based on solving a QP online. Such a scheme would be endowed with stability guarantees. Alternatively, by replacing the objective with $\|u^{\mathrm{r}}_t-u\|_2^2$ where $u^{\mathrm{r}}_t$ is a reference input signal, safety-filter type schemes \cite{Wabersich2021} can be synthesized.

In Figure \ref{fig:control_figure}, we plot the closed-loop trajectories with $u_t=u_t^*$, where the reference is composed of piecewise constant signals, and both the observer and plant are initialized at the origin $-$ Since the structure of the plant is known, it can reasonably be assumed that $\hat{x}_0-z_0 \in \mathcal{E}$ is verified. We observe that the plant output tracks the reference while being included in the output constraint set $\mathbb{Y}$. We observe further that $y_t$ and $Cz_t$, where $z_t$ is the observer state, are very close to each other. This indicates that our choice of disturbance set satisfies Assumptions \ref{ass:W_assumption}, such that the controller is able to perform safe tracking following Proposition \ref{prop:RCI_exists}. In Figure \ref{fig:RCI_figure}, we plot the set $\hat{\mathbb{X}}$ computed as the convex hull of the sets $\mathbb{X}(q_t^*)$ computed by Problem \eqref{eq:feedback_controller} over simulation length of $20$s.
We report that $z_t \in \mathbb{X}(q^*_t)$ always holds resulting in $z_t \in \hat{\mathbb{X}}$, demonstrating recursive feasibility. 

\begin{figure}
	\centering
	\includegraphics[width=1.\linewidth, clip=true]{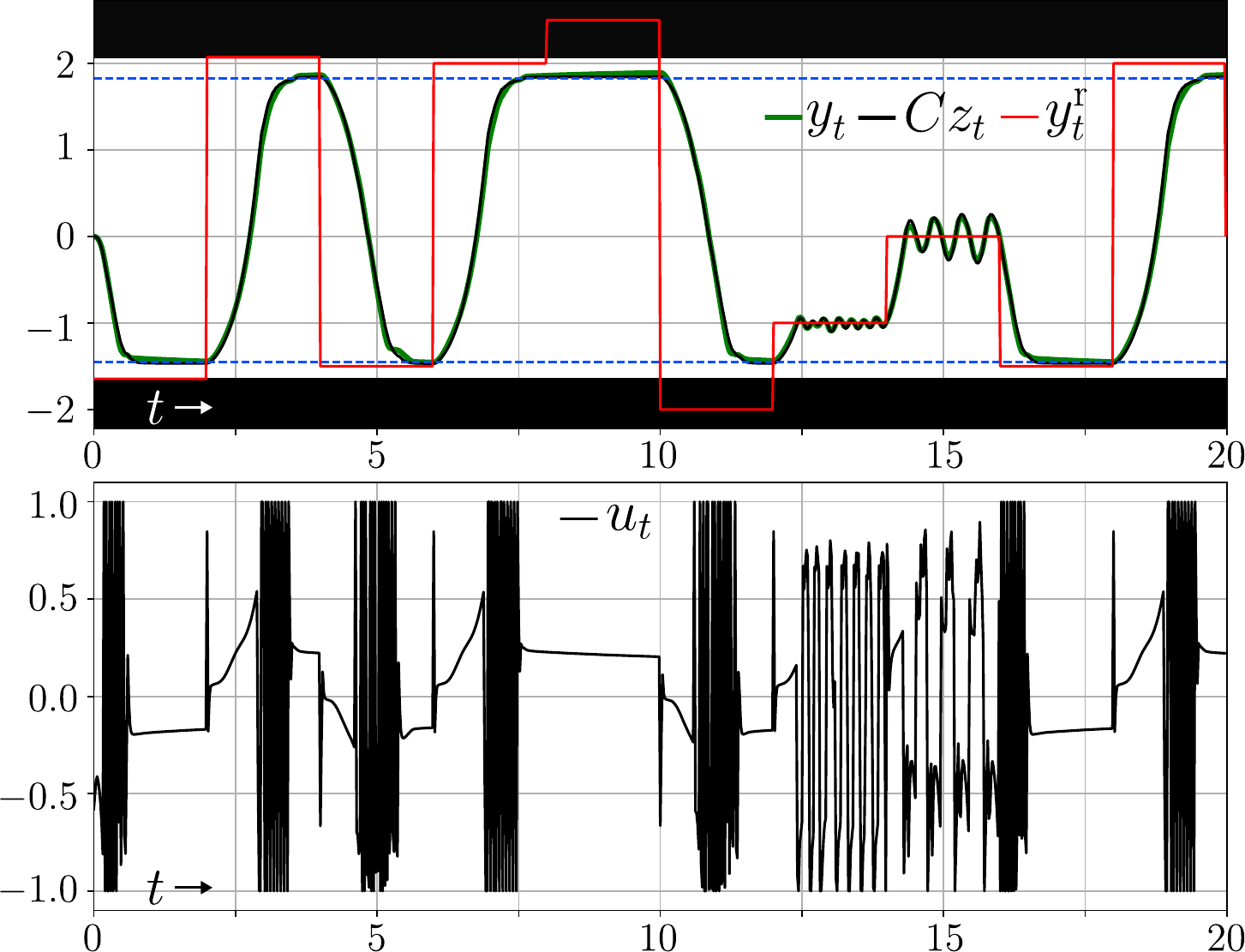}
	\captionsetup{width=\linewidth,justification=justified, singlelinecheck=false}
	\caption{(Top) Closed-loop trajectories of plant and observer in the output space. Observe that the plant satisfies constraints, indicated by the black regions. The black line indicates boundaries of the tightened output constraint set $\mathbb{Y} \ominus \mathbb{W}_{\kappa}(\mathrm{c}_{\mathrm{w}},\epsilon_{\mathrm{w}})$; (Bottom) Corresponding input trajectories generated by \eqref{eq:feedback_controller}.}
	\label{fig:control_figure}
\end{figure}

\begin{figure}
	\centering
	\includegraphics[width=1.\linewidth, clip=true]{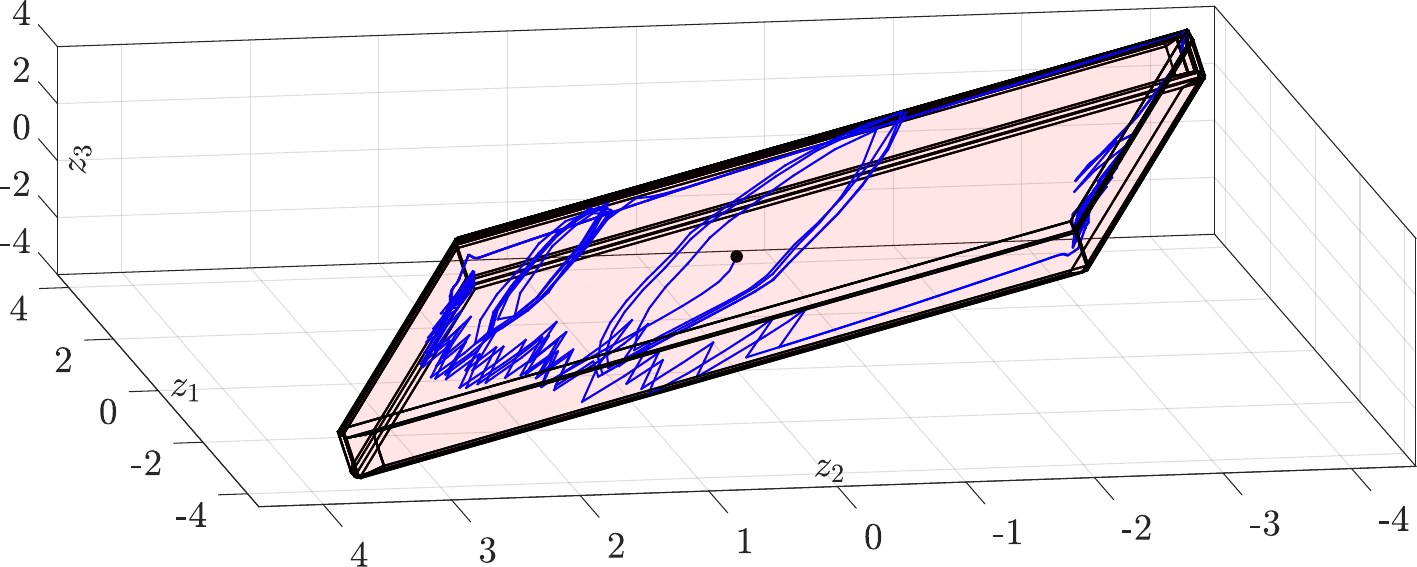}
	\captionsetup{width=\linewidth,justification=justified, singlelinecheck=false}
	\caption{State trajectories of the closed-loop observer \eqref{eq:closed_loop_observer}, along with the set $\hat{\mathbb{X}}$ to demonstrate recursive feasibility of Problem \eqref{eq:feedback_controller}.}
	\label{fig:RCI_figure}
\end{figure}

\section{Conclusions}
By properly defining a control-oriented regularization term into the system identification problem,
we presented an approach to guarantee that the identified model is suitable for designing constrained controllers 
for the plant generating the data. An uncertain model was derived using a state observer, and the regularization function was characterized as the size of the largest RCI set for the uncertain model. By parameterizing the model as a qLPV model and representing the RCI set as a configuration-constrained polytope, we transformed the identification problem into a computable form. The initialization strategy, which includes estimating suitable number of scheduling signals, enables solving the problem
to good-quality solutions as numerically demonstrated.

The results of the paper can be extended in several directions, such as: (a) reduce conservativeness in the RCI set by limiting the multiplicative uncertainty encountered when the system evolves within the set; (b) develop alternative formulations of $\mathbf{r}(\Theta)$; (c) design tailored optimization algorithms; and (d) integrate the scheme with an active learning framework for data-efficient control-oriented model identification.

%\section*{References}
\bibliographystyle{ieeetr} % Choose a bibliography style
\bibliography{manuscript_references} % references.bib is the name of the .bib file

\end{document}